\pdfoutput=1
\documentclass[a4paper,UKenglish,cleveref, autoref, thm-restate]{lipics-v2019}

\usepackage{hyperref}
\usepackage{amssymb, amsmath}
\usepackage{amsthm}
\usepackage{dsfont}
\usepackage[ruled,vlined,linesnumbered]{algorithm2e}
\usepackage{tikz}
\usetikzlibrary{shapes,backgrounds}

\newcommand{\val}{\mbox{Val}}
\newcommand{\M}{\text{\sc max} }
\newcommand{\m}{\text{\sc min} }
\newcommand\comments[1]{}
\newcommand{\pro}[1]{{\mathbb P}\left(#1\right)}
\newcommand{\pros}[3]{{\mathbb P}^{#1}_{#2}\left(#3\right)}
\newcommand{\esp}[1]{{\mathbb E}\left(#1\right)}
\newcommand{\esps}[3]{{\mathbb E}^{#1}_{#2}\left(#3\right)}   
 
\everymath{\displaystyle}

\DeclareMathOperator*{\argmax}{arg\,max}

\title{A Generic Strategy Improvement Method for Simple Stochastic Games} 

\author{David Auger}{Université Paris Saclay, UVSQ, DAVID}{david.auger@uvsq.fr}{}{}

\author{Xavier Badin de Montjoye}{Université Paris Saclay, UVSQ, DAVID}{xavier.badin-de-montjoye2@uvsq.fr}{}{}

\author{Yann Strozecki}{Université Paris Saclay, UVSQ, DAVID}{yann.strozecki@uvsq.fr}{}{}

\authorrunning{D. Auger and X. Badin De Montjoye and Y. Strozecki} 

\Copyright{D. Auger and X. Badin de Montjoye and Y. Strozecki} 



\acknowledgements{The authors want to thank Pierre Coucheney for many interesting discussions on SSGs.}

\nolinenumbers 

\hideLIPIcs  

\EventEditors{John Q. Open and Joan R. Access}
\EventNoEds{2}
\EventLongTitle{42nd Conference on Very Important Topics (CVIT 2016)}
\EventShortTitle{CVIT 2016}
\EventAcronym{CVIT}
\EventYear{2016}
\EventDate{December 24--27, 2016}
\EventLocation{Little Whinging, United Kingdom}
\EventLogo{}
\SeriesVolume{42}
\ArticleNo{23}

\begin{document}

\maketitle

\begin{abstract}
    We present a generic strategy improvement algorithm (GSIA) to find an optimal strategy of simple stochastic games (SSG). We prove the correctness of GSIA, and derive a general complexity bound, which implies and improves on the results of several articles. First, we remove the assumption that the SSG is stopping, which is usually obtained by a polynomial blowup of the game. Second, we prove a tight bound on the denominator of the values associated to a strategy, and use it to prove that all strategy improvement algorithms are in fact fixed parameter tractable in the number $r$ of random vertices. All known strategy improvement algorithms can be seen as instances of GSIA, which allows to analyze the complexity of converge from below by Condon~\cite{condon1993algorithms} and to propose a class of algorithms generalising Gimbert and Horn's algorithm~\cite{gimbert2008simple,lmcs1119}. These algorithms
    terminate in at most $r!$ iterations, and for binary SSGs, they do less iterations than the current best deterministic algorithm  given by Ibsen-Jensen and Miltersen~\cite{ibsen2012solving}.
\end{abstract}


\section{Introduction}

A \emph{simple stochastic game}, or SSG, is a two-player turn-based zero-sum game with perfect information introduced by Condon~\cite{condon1992complexity}. It is a simpler version of \emph{stochastic games}, previously defined by Shapley~\cite{Shapley1095}. An SSG is played by two players {\sc max} and {\sc min} moving a pebble on a graph. Vertices are divided into \m vertices, \M vertices, random vertices and a target vertex for {\sc max}. When the pebble reaches a \m or \M vertex,  corresponding players move the pebble to a neighbouring vertex of their choice. If it reaches a random vertex, the next vertex is chosen at random following some probability law. Finally, when the pebble reaches the target vertex, \m pays $1$ to \M. The goal of \m is to minimise the probability to reach the target vertex while \M must maximise this probability.

We study the algorithmic problem of {\it solving} an SSG, i.e. finding a pair of optimal strategies in an SSG,
or equivalently the optimal value vector of the optimal probabilities for \M to reach the sink from each vertex. There are always optimal strategies for both players that are positional~\cite{condon1992complexity}, i.e. stationary and deterministic, but the number of positional strategies is exponential in the size of the game. Consequently, finding a pair of optimal  strategies is a problem not known to be in $\textsc{FP}$, but it is in $\textsc{PPAD}$~\cite{juba2005hardness}, a class included in $\textsc{FNP}$.


Simple Stochastic Games can be used to simulate many classical games such as parity games, mean or discounted payoff games~\cite{Anderson2009,chatterjee2011parity}. Moreover, stochastic versions of these games are equivalent to SSGs~\cite{Anderson2009}, which underlines that SSGs are an important model to study. SSGs have applications in different domains such as model checking of modal $\mu$-calculus~\cite{stirling1999bisimulation}, or modelling  autonomous urban driving~\cite{Chen2013urbandriving}.


There are three known methods to solve SSGs: strategy improvement, value iteration and quadratic programming. A strategy improvement algorithm (SIA) starts with a strategy for one player and improves it until it is optimal, whereas value iteration algorithms (VIA) update a value vector by elementary operations, which converges to the optimal value vector of the game. Implementations of those algorithms have been written and compared in~\cite{kvretinsky2020comparison}. 

Denote by $n$ be the number of \M vertices and $r$ be the number of random vertices in an SSG.
For SSGs with \M  vertices of outdegree $2$, the best known deterministic algorithm is an SIA which makes at worst $\displaystyle O\left(2^{n}/n\right)$ iterations (see~\cite{tripathi2011strategy}), and the best known randomised algorithm is a SIA described by Ludwig in~\cite{ludwig1995subexponential}, which runs in $\displaystyle 2^{O\left(\sqrt{n}\right)}$. 

 Gimbert and Horn give an SIA in~\cite{gimbert2008simple}, running in $\displaystyle O^*\left(r!\right)$ iterations, namely a superpolynomial dependency in $r$ only ($O^*$ omits polynomial factors in $r$ and $n$). For SSGs where random vertices have a probability distribution $(1/2,1/2)$ (coin toss), Ibsen-Jensen and Miltersen present a VIA of complexity in $\displaystyle O^*\left(2^{r}\right)$~\cite{ibsen2012solving}. It turns out that all SIA runs in $\displaystyle O^*\left(2^{r}\right)$ on SSGs with probability distribution $(1/2,1/2)$, as we prove in this article. The same complexity of $O^*\left(2^{r}\right)$ is obtained for general SSGs with a more involved randomised algorithm in~\cite{auger2019solving}.
 
 Most of the aforementioned algorithms rely on the game being stopping, meaning that it structurally ends in a sink with probability $1$. This condition is not restrictive since any SSG can be transformed into a stopping SSG while keeping the same optimal strategies. However, this transformation incurs a quadratic blow-up of the game and cannot be used in real life application. In this paper we give bounds in $O\left(2^{r}Poly(n)\right)$ computational time for some kinds of algorithm. Using the stopping transformation would have induced a $O\left(2^{nr^2}Poly(n)\right)$ complexity. The stopping restriction has been lifted for quadratic programming in~\cite{kvretinsky2020comparison} and before that for SIA and VIA in~\cite{Chatterjee2008valueiteration,CHATTERJEE2013reachability}.

\subsection*{Contributions}

We introduce GSIA, a new meta-algorithm to solve SSGs in Sec.~\ref{sec:GIS}. This algorithm proves simultaneously the correctness of multiple algorithms (\cite{condon1993algorithms,gimbert2008simple,dai2009new, tripathi2011strategy,ibsen2012solving,auger2019solving}). In Sec.~\ref{sec:complexity}, we give a general complexity bound that matches or improves on previous bounds obtained by ad-hoc methods. We show that all these algorithms are fixed-parameter tractable in the number of random vertices. Moreover, we do not rely on the fact that the game is stopping, which was commonly used in the aforementioned papers.
The proof of correctness  relies on a notion of concatenation for strategies and an analysis of absorbing sets in the game, while the complexity bound is derived from a new and tight characterisation of the values of an SSG. Finally, in Sec.~\ref{sec:GSIinstances}, we show how GSIA can be used to derive new algorithms, generalising classical ones. In particular, we exhibit a class of algorithms which generalise Gimbert and Horn's algorithm
and use less iterations than Ibsen-Jensen and Miltersen's algorithm.

We emphasise that our goal here is not to define a new
algorithm that would have a better --but still exponential-- complexity bound than the state of the art (a sub-exponential algorithm for SSGs, like the ones found for parity games \cite{calude2020deciding, colcombet2019universal}, would already be a significant improvement),
but rather to wrap-up a lot of previous research
by showing that all known SIA for SSGs, despite having emerged in different contexts and having ad-hoc proofs of convergence, are in fact instances of a general pattern that can be
further expanded, and actually share the
best known complexity bounds.


\section{Simple Properties of Simple Stochastic Games}\label{sec:definitions}


    \subsection{Simple Stochastic Games}

       We give a generalised definition of \emph{Simple Stochastic Game}, a two-player zero-sum game with turn-based moves and perfect information introduced by Anne Condon~\cite{condon1993algorithms}. 
        
        \begin{definition}
            A Simple Stochastic Game (SSG) is a directed graph $G$, together with:
            \begin{enumerate}
                \item A partition of the vertex set $V$ in four parts $V_{\M}$, $V_{\m}$, $V_R$ and $V_S$ (all possibly empty, except $V_S$), satisfying the following conditions:
                \begin{enumerate}
                    \item every vertex of $V_{\M}$, $V_{\m}$ or $V_R$ has at least one outgoing arc;
                    \item every vertex of $V_S$ has exactly one outgoing arc which is a loop on itself.
                \end{enumerate}
                \item For every $x\in V_R$, a probability distribution $p_x(\cdot)$ with rational values, on the outneighbourhood of $x$.
                \item For every $x \in V_S$, a value $\val(x)$ which is a rational number in the closed interval $[0,1]$.
            \end{enumerate}
        \end{definition}
        
        In the article, we denote $|V_{\M}|$ by $n$ and $|V_R|$ by $r$. 
        Vertices from $V_{\M}$, $V_{\m}$, $V_R$ and $V_S$ are respectively called \M vertices, \m vertices, random vertices and sinks. For $x \in V$, we denote by $N^{+}(x)$ the set of outneighbours of $x$. We assume that for every $x \in V_R$ and $y \in V$, $y \in N^{+}(x)$ if and only if $p_x(y) > 0$.
        
        The game is played as follows. The two players are named \M and {\sc min}. A token is positioned on a starting vertex $x$. If $x$ is in $V_{\M}$ (resp. $V_{\m}$) the \M player (resp. the \m player) chooses one of the outneighbours of $x$ to move the token to. If $x$ is in $V_R$, the token is randomly moved to one of the outneighbours of $x$ according to the probability distribution $p_x(\cdot)$, independently of everything else. This process continues until the token reaches a sink $s$ and then, player \m has to pay $\val(s)$ to player \M and the game stops. 
        The problem we study is to \emph{find the best possible strategies} for \m and {\sc max}, and the expected value that \m has to pay to \M while following those strategies. 
        
    
        We consider a slightly restricted class of SSGs where the probability distribution on each random vertex has a given precision and the value of the sinks are $0$ and $1$.
        
        \begin{definition}
        
            For $q$ a positive integer, we say that an SSG is a $q$-SSG if there are only two sinks of value $0$ and $1$, and for all $x \in V_{R}$, there is an integer $q_x \leq q$ such that the probability distribution $p_x(\cdot)$ can be written as $p_x(x') = \displaystyle \frac{\ell_{x,x'}}{q_x}$ for all $x'$ where $\ell_{x,x'}$ is a natural number.
        
        \end{definition}
        
        As an example, let $x$ be a random vertex of a $2$-SSG, and let $u \in N^+(x)$, then $p_x(u)$ can be equal to $0$, $1/2$ or $1$. The case $p_x(u) = 0$ is forbidden by definition, and if $p_x(u) = 1$, then $x$ is of degree one and can be removed (by redirecting arcs entering $x$ directly to $u$), without changing anything about the outcome of the game.
        Hence, we suppose without loss of generality that each random vertex of a $2$-SSG has degree $2$ and has probability distribution $(1/2,1/2)$. This definition matches the one of a {\it binary SSG}, given by Condon and used in most articles on SSGs, except that we allow here \M and \m vertices to have an outdegree larger than $2$.
    
    \subsection{Play, History and Strategies}
    
        \begin{definition}
            A {\bf play} in $G$ is an infinite sequence of vertices
            $X= (x_0, x_1, x_2, \cdots)$ such that  for all $t \geq 0$,
            $(x_t,x_{t+1})$ is an arc of $G$. 
        \end{definition}
        
        If for a play $X=(x_t)$ there is some $t \geq 0$ with $x_t = s \in V_S$, then all subsequent vertices in the play are also equal to $s$. In this case, we say that the play { \bf reaches} sink vertex $s$ and we define the {\bf value of the play} $\val(X)$ as $\val(s)$. If the play reaches no sink, then we set $\val(X) = 0$.
        
        A {\bf history} of $G$ is a finite directed path $h = (x_0, x_1, \cdots, x_k)$. If the last vertex $x_k$ is a \M vertex (resp. \m vertex), we say that $h$ is a \M history (resp. \m history).
        
        \begin{definition}
            A {\bf general \M strategy} (resp. general \m strategy) is a map $\sigma$ assigning to every \M history
            (resp. \m history) $h=(x_0, x_1, \cdots, x_k)$ a vertex $\sigma(h)$ which is an outneighbour of $x_k$. The set of these strategies is denoted by $\Sigma_{gen}^{\M}$ (resp. $\Sigma_{gen}^{\m}$).        \end{definition}
        
        For $\sigma \in \Sigma_{gen}^{\M}$ and $\tau \in \Sigma_{gen}^{\m}$, given a starting vertex $x_0$, we recursively define a random play $X = (X_0, X_1, \cdots)$ of $G$ in the following way. At $t=0$ let $X_0 = x_0$, and for $t \geq 0$:
        
        \begin{itemize}
            \item if $X_t \in V_{\M}$, define $X_{t+1} = \sigma(X_0, X_1, \cdots, X_t)$;
            \item if $X_t \in V_{\m}$, define $X_{t+1} = \tau(X_0, X_1, \cdots, X_t)$;
            \item if $X_t \in V_{R}$, then $X_{t+1}$ is an outneighbour of $X_t$ chosen following the probability distribution $p_{X_t}(\cdot)$, independently of everything else;
            \item if $X_t \in V_{S}$, define  $X_{t+1} = X_t$.
        \end{itemize}
        
        This defines a distribution on plays which we denote by $\pros{x_0}{\sigma, \tau}{\cdot}$, or simply $\pro{\cdot}$ if strategies and starting vertex are clear from context. 
        The corresponding expected value and conditional expected values are denoted by $\esps{x_0}{\sigma, \tau}{\cdot \vert \cdot}$, or simply $\esp{\cdot \vert \cdot}$.
        
        We now define \textbf{positional strategies} which only depend on the last vertex in the history:
        
        \begin{definition}
        
            A general \M strategy $\sigma$ (resp. \m strategy) is said to be positional if for any \M vertex $x$ (resp. \m vertex) and any history $h = (x_0, \ldots, x)$, we have $\sigma(h) = \sigma((x))$ where $(x)$
            is the history containing only $x$ as a start vertex. The set of positional \M strategies (resp. \m strategies) is denoted $\Sigma^{\M}$ (resp. $\Sigma^{\m}$).
        
        \end{definition}

    \subsection{Values in an SSG}
    
        \begin{definition}
        
            Let $G$ be an SSG and let $(\sigma,\tau)$ be a pair of \M and \m strategies, the value vector $v^{G}_{\sigma,\tau}$ is the real vector of dimension $|V|$ defined by, for any $x_0 \in V$, 
            $$v^{G}_{\sigma,\tau}(x_0) = \esps{x_0}{\sigma, \tau}{\val(X)}.$$
            
            This value represents the expected gains for player \M if both players plays according to $(\sigma,\tau)$ and the game starts in vertex $x_0$.
             
        \end{definition}
        As before, the superscript $G$ can be omitted when the context is clear.
        
        To compare value vectors, we use the pointwise order:  we say that $v \geq v'$ if for all vertices $x \in V$ we have $v(x) \geq v'(x)$. Moreover, we say that $v > v'$ if $v \geq v'$ and there is some $x$ such that  $v(x) > v'(x)$.
        Given a $\M$ strategy $\sigma$, a {\bf best response} to $\sigma$
        is a $\m$ strategy $\tau$ such that $v_{\sigma,\tau} \leq v_{\sigma,\tau'}$
        for all $\m$ strategies $\tau'$.
       
        
        \begin{proposition}[\cite{condon1993algorithms}]
        A positional strategy admits a positional best response, which can be found in polynomial time using linear programming.
        \end{proposition}
        
         The set of positional best responses to $\sigma$ is denoted by $BR(\sigma)$. Similarly, for a \m strategy $\tau$, we define the notion of best response to $\tau$ and the corresponding set is denoted by $BR(\tau)$. Except explicitly stated otherwise (in Sec. \ref{seq:concat}), all considered strategies are positional.

        We denote by $\tau(\sigma)$ a positional best response to $\sigma$.
        For a \M strategy $\sigma$ and $\tau \in BR(\sigma)$, we write $v_{\sigma}$ for $v_{\sigma,\tau}$. 
        For a \m strategy $\tau$ and $\sigma \in BR(\tau)$, we write $v_{\tau}$ for $v_{\sigma,\tau}$. The vector $v_\sigma$ is called the value vector of strategy $\sigma$, and is used to compare strategies by writing $\sigma' \underset{G}{>} \sigma$ if and only if $v^{G}_{\sigma'} > v^{G}_{\sigma}$. 
  
        
        
        
        It is well known (see \cite{condon1993algorithms,tripathi2011strategy}) that there is a pair of deterministic positional strategies $(\sigma^*,\tau^*)$ called optimal strategies, that satisfies for all $x$,
        $v^* = v_{\sigma^*,\tau^*} = v_{\sigma^*} = v_{\tau^*}$ since $\sigma^*$ and $\tau^*$ are best responses to each other. 



    \subsection{Optimality Conditions}
        
        
        The next two lemmas give characterisations of (optimal) value vectors under a pair of strategies. They are fundamental to all algorithms finding optimal strategies. Proofs of similar results can be found in~\cite{condon1992complexity}; we add here a fifth condition to make the characterisation hold when the game is not stopping.
        
        For any SSG, the vertices with value $0$ under optimal strategies can be found in linear time by a simple graph traversal computing its complementary, the set of \M vertices which can access a sink of positive value, regardless of the choice of the \m player. Let $K^{G}$ be the set of vertices with value$0$ under optimal strategies. For a \M strategy $\sigma$ of $G$ and a \m strategy $\tau$, we call $K^{G}_{\sigma,\tau}$ the set of vertices with value zero under the pair of strategies $\sigma,\tau$, and $K^{G}_{\sigma}$ the set of vertices with value zero under $\sigma,\tau$ when $\tau$ is a best response to $\sigma$. 
        
        \begin{lemma}\label{lemma:NSOptConditions}
            Given positional strategies $(\sigma,\tau)$ and a real $|V|$-dimensional vector $v$, one has
            equality between $v$ and $v_{\sigma,\tau}$ if and only if the following conditions are met:
            
            \begin{romanenumerate}
                \item For $s \in V_S$, $v(s) = \val(s)$
                \item For $r \in V_R$, $v(r) = \sum\limits_{y \in N^+(r)} p_r(y)v(y)$
                \item For $x \in V_{\m}$, $v(x) = v(\tau(x))$
                \item For $x \in V_{\M}$, $v(x) = v(\sigma(x))$
                \item For any $x \in V$, $v(x) = 0$, if and only if $x \in K^{G}_{\sigma,\tau}$
            \end{romanenumerate}
            
            Moreover, $\tau \in BR(\sigma)$ if and only if for any $x$ in $V_{\m}$, $v(x) = \min\limits_{y \in N^{+}(x)} v(y) = v(\tau(x))$ and the last condition is modified into $v(x) = 0$ if and only if $x \in K^{G}_{\sigma}$.
        \end{lemma}
        
        
        \begin{lemma}[Optimality conditions]\label{lemma:optimality_condition}
           Given positional strategies $(\sigma,\tau)$ and denoting $v = v_{\sigma,\tau}$, $(\sigma,\tau)$ are optimal strategies if and only if:
            \begin{romanenumerate}
                \item For $s \in V_S$, $v(s) = \val(s)$
                \item For $r \in V_R$, $v(r) = \sum\limits_{y \in N^+(r)} p_r(y)v(y)$
                \item For $x \in V_{\m}$, $v(x) = \min\limits_{y \in N^{+}(x)} v(y)$
                \item For $x \in V_{\M}$, $v(x) = \max\limits_{y \in N^{+}(x)} v(y)$
                \item For any $x \in V$, $v(x) = 0$, if and only if $x \in K^{G}$
            \end{romanenumerate}
        \end{lemma}
        
        The conditions of Lemma~\ref{lemma:optimality_condition} imply that $(\sigma,\tau)$ is a certificate of optimality that can be checked in polynomial time: compute $v_{\sigma,\tau}$ by solving the linear system of Lemma~\ref{lemma:NSOptConditions}, compute  $K^{G}$ in linear time, then check in linear time if conditions are met.
    
\section{Generic Strategy Improvement Algorithm}\label{sec:GIS}

    \subsection{Game Transformation}
    
        We present a simple transformation of an SSG, where some arcs of the game are rerouted to new sinks with appropriate values.

        \begin{definition}
        
            Let $G$ be an SSG, $A$ be a subset of the arcs of $G$ and $f$ be a function from $A$ to the set of rational numbers. Let $G[A,f]$ be the SSG obtained from a copy of $G$ with the following modifications:  each arc $e = (x,y) \in A$ is removed and replaced in $G[A,f]$ by $e'= (x,s_e)$
            where $s_e$ is a new sink vertex with value $f(e)$.
            These new sinks of $G[A,f]$  are called $A$-sinks, and $A$ is called the set of {\it fixed arcs}.
        \end{definition}

        \begin{figure}
\centering

    \begin{tikzpicture}
    
        \node[draw, circle] (1a) at (0,1) {$x_1$};
		\node[draw, circle] (2a) at (0,-1) {$x_2$};
		\node[draw, circle] (3a) at (1.5,0) {$x_3$};
		\node[draw, circle] (5a) at (3,1) {$x_4$};
		\node[draw, circle] (6a) at (3,-1) {$x_5$};
		\node (tr) at (4.5,0) {$\longrightarrow$};
		\node[draw, circle] (1b) at (6,1) {$x_1$};
		\node[draw, circle] (2b) at (6,-1) {$x_2$};
		\node[draw, circle] (3b) at (7.5,0.5) {$x_3$};
		\node[draw, circle, fill=gray!50] (4b) at (7.5,-0.5) {$0.3$};
		\node[draw, circle] (5b) at (9,1.5) {$x_4$};
		\node[draw, circle] (6b) at (9,-0.5) {$x_5$};
		
		\draw[->,>=latex] (1a)--(3a);
		\draw[->,>=latex, dashed] (2a)--(3a);
		\draw[->,>=latex] (3a)--(5a);
		\draw[->,>=latex] (3a)--(6a);
		\draw[->,>=latex] (2a) to[bend right = 30] (6a);
		\draw[->,>=latex] (1b)--(3b);
		\draw[->,>=latex] (2b)--(4b);
		\draw[->,>=latex] (3b)--(5b);
		\draw[->,>=latex] (3b)--(6b);
		\draw[->,>=latex] (2b) to[bend right = 30] (6b);
    
    \end{tikzpicture}

\caption{Transformation of the graph $G$ in $G[\{(x_2,x_3)\},f]$ where $f((x_2,x_3))=0.3$}
\label{figure:definition}
\end{figure}
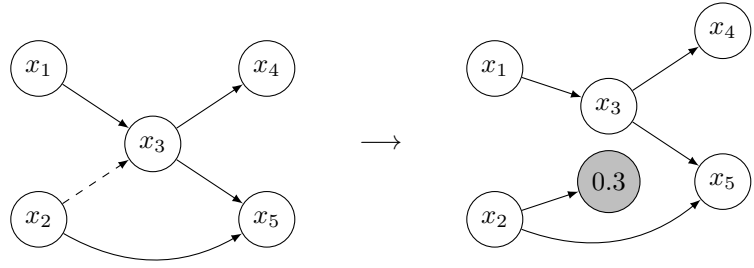
        
        Note that in the previous definition, the end vertex $y$
        of an arc $(x,y) \in A$ is not removed from the game. Its incoming
        arcs which are in $A$ are simply redirected to sinks, see Fig.~\ref{figure:definition}.
        
        The function $f$ is usually given by the values of a a strategy: we denote by $G[A,\sigma]$  the game $G[A,f]$, where $f$ is defined on every arc $e = (x,y)$ of $A$ by $f(e)=v_{\sigma}(y)$.
        Comparing $G$ and $G[A,\sigma]$, the only differences are that arcs of $A$ have their endpoints changed to new sinks. Therefore, a strategy defined in $G$ can be interpreted as a strategy of $G[A,\sigma]$ and vice versa, and we identify strategies in $G$ and $G[A,\sigma]$. However, when we compare 
        the values of a strategy in both games (as in Lemma \ref{NonChangement} below), it makes sense to compare
        only the values on vertices in $G$ and not on $A$-sinks (and anyway values of $A$-sinks are fixed).
        

        \begin{lemma}
        
            \label{lemma:ZeroMeme}
            
            For an SSG $G$, a subset of arcs $A$, and a \M strategy $\sigma$, $K^{G}_{\sigma} = K^{G[A,\sigma]}_{\sigma}$.
        
        \end{lemma}
        
        \begin{proof}
            Fix a min strategy $\tau$ and
            define $R_{\sigma,\tau}^G(x)$  
            as the set of vertices that can be reached 
            from $x$ in $G$, following only arcs corresponding
            to $\sigma$ and $\tau$ after $\M$ and $\m$ vertices, and any arc out of random vertices.
            We repeatedly use the easy fact that the three following assertions
            are equivalent:
            \begin{enumerate}[(i)]
            \item $v^G_{\sigma,\tau}(x) = 0$;
            \item $v^G_{\sigma,\tau}(y) = 0$ for all $y \in R^G_{\sigma,\tau}(x)$;
            \item $\val^G(s)=0$ for all $s \in V^G_S \cap R^G_{\sigma,\tau}(x)$.
            \end{enumerate}
            The same equivalence is true in $G[A,\sigma]$, where we define $R^{G[A,\sigma]}_{\sigma,\tau}$ likewise.
            Denote by $R^G_A(x)$ vertices 
            of $R^G_{\sigma,\tau}(x)$ that are endpoints of
            arcs in $A$, and let $S_A(x)$ be the corresponding
            $A$-sinks in $G[A,\sigma]$.
    
            Suppose that $v^{G}_{\sigma,\tau}(x) = 0$ and 
            consider a sink $s$ in $V_S^{G[A,\sigma]} \cap R^{G[A,\sigma]}(x)$:
            either it belongs to $V^G_S$ hence also
            to $R^{G}(x)$ and satisfies $\val^G(s)=0$
            by $(iii)$, or it belongs to $S_A(x)$ and
            then by definition
            $$\val^{G[A,\sigma]}(s) = v^G_\sigma(s) \leq v^G_{\sigma,\tau}(s) = 0.$$
            Thus, by $(iii)$ once again we have $v^{G[A,\sigma]}_{\sigma,\tau}(x)=0$.
            
            Conversely, suppose that $v^{G[A,\sigma]}_{\sigma,\tau}(x) = 0$ and let $s \in V^G_S \cap R^G_{\sigma,\tau}(x)$. 
            Then, either
            $s \in R^{G[A,\sigma]}_{\sigma,\tau}$, hence by $(iii)$
            $$\val^G(s) = \val^{G[A,\sigma]}(s) = 0,$$
            or there is a $y \in R_A^G(x)$ such
            that $s \in R^G_{\sigma,\tau}(y)$. In this case we have  $v^{G}_{\sigma,\tau}(y)=0$ by $(ii)$,
            hence $\val^G(s)=0$ by $(iii)$ applied to $y$,
            and we see that $v^{G}_{\sigma,\tau}(x) = 0$.
            
            Since we have $v^{G}_{\sigma,\tau}(x) = 0$ if and only if 
            $v^{G[A,\sigma]}_{\sigma,\tau}(x) = 0$, regardless of $\tau$, the result follows.
        \end{proof}
        
        \begin{lemma}
        
            \label{NonChangement}
        
            For an SSG $G$, a subset of arcs $A$, and a \M strategy $\sigma$, $v^{G}_{\sigma} = v^{G[A,\sigma]}_{\sigma}$.
        
        \end{lemma}

        \begin{proof}
            This is a direct consequence of Lemma~\ref{lemma:NSOptConditions} and Lemma~\ref{lemma:ZeroMeme},
            since the vector $v^{G}_{\sigma}$ satisfies the best-response conditions in $G[A,\sigma]$ and vice versa.
        \end{proof}

    \subsection{The Algorithm}
    
        An SSG is \textbf{stopping} if under every pair of strategies, a play eventually reaches a sink with probability $1$. Most algorithms in the literature depend on the game being stopping.
        It is usually not seen as a limitation since it is possible to transform every SSG into a stopping SSG, but the transformation makes the game polynomially larger by adding $O(nr)$ random vertices, which is bad from a complexity point of view, especially for algorithm with parametrized complexity in the number of random vertices.
        We strengthen the classical order on strategies, to get rid of the stopping condition in the generic strategy improvement algorithm presented in this section.

        \begin{definition}
            Let $\sigma$ and $\sigma'$ be two \M strategies, then $\sigma' \underset{G}{\succ} \sigma$ if $\sigma' \underset{G}{>} \sigma$ and for every \M vertex $x$, if $v^{G}_{\sigma'}(x) = v^{G}_{\sigma}(x)$,  then $\sigma'(x) = \sigma(x)$.
        \end{definition}
        
     Algorithm~\ref{MA} is a classical strategy improvement algorithm with two twists: 
     the improvement is for \emph{the stricter order $\succ$} and it is guaranteed in  \emph{the transformed game} rather than in the original game. We call Algorithm~\ref{MA} the Generic Strategy Improvement Algorithm, or~GSIA.

        \begin{algorithm}
        	\caption{GSIA\label{MA}}
        	\DontPrintSemicolon
        	\KwData{$G$ a stopping SSG}
        	\KwResult{$(\sigma,\tau)$ a pair of optimal strategies}
        	\Begin{
        	    select an initial \M strategy $\sigma$\;
        	    \While{$(\sigma,\tau(\sigma))$ are not optimal strategies of $G$}{
        	        choose a subset $A$ of arcs of $G$\;
        	        find $\sigma'$ such that $\sigma' \underset{G[A,\sigma]}{\succ} \sigma$.\;
        	        $\sigma \longleftarrow \sigma'$\;
        	    }
        	    
        	    \KwRet{$(\sigma,\tau(\sigma))$}
        	}
        \end{algorithm}
        
        Algorithm~\ref{MA} is a generic algorithm (or meta-algorithm) because neither the selection of an initial strategy $\sigma$ at line $2$, nor the way of choosing $A$ at line $4$, nor the way of finding $\sigma'$ at line~$5$, are specified. A choice of implementation for these three parts is an \textbf{instance} of GSIA, that is a concrete strategy improvement algorithm.
        Note that if $\sigma' \underset{G[A,\sigma]}{>} \sigma$ is found,
        it is easy to find $\sigma''$ with $\sigma'' \underset{G[A,\sigma]}{\succ} \sigma$: define $\sigma''$ as equal to $\sigma'$, except for $\M$ vertices $x$ such that $v^{G}_{\sigma'}(x) = v^{G}_{\sigma}(x)$ and $\sigma'(x) \neq \sigma(x)$ where $\sigma''(x)$ is defined as $\sigma(x)$.
        
        When we prove some property of GSIA in this article, it means that the property is true for all instances of GSIA, that is regardless of the selection of the initial strategy, the set $A$ and the method for selecting  $\sigma'$. 
    
        In order to prove the correctness of GSIA, we need to prove two points:
        \begin{enumerate}
            \item If $\sigma$ is not optimal in $G$, then $\sigma$ is not optimal in $G[A,\sigma]$.
            \item If $\sigma' \underset{G[A,\sigma]}{\succ} \sigma$ then $\sigma' \underset{G}{>} \sigma$. 
        \end{enumerate}

        The first point is proved in the following lemma, while the second one is harder to obtain and is the subject of the next two subsections.
        
        \begin{lemma}\label{lemma:optimality_transfer}
        
            For an SSG $G$ and a subset of arcs $A$, a \M strategy $\sigma$ is optimal in $G$ if and only if it is optimal in $G[A,\sigma]$.
        
        \end{lemma}
        
        \begin{proof}
        
            Except on $A$-sinks, the value vectors of $\sigma$ in $G$ and $G[A,\sigma]$ are equal by Lemma~\ref{NonChangement}. Furthermore, by Lemma~\ref{lemma:ZeroMeme},
            $K^{G}_{\sigma} = K^{G[A,\sigma]}_{\sigma}$;
            hence $\sigma$ satisfies the optimality conditions of Lemma~\ref{lemma:optimality_condition} in $G$ if and only if it satisfies them in $G[A,\sigma]$.
        \end{proof}
        
    \subsection{Concatenation of Strategies} \label{seq:concat}
        
        As a tool for proving the correctness of Algorithm~\ref{MA}, we introduce the notion of concatenation of strategies which produces non-positional strategies even if both concatenated strategies are positional. The idea of using a sequence of concatenated strategies to interpolate between two strategies has been introduced in~\cite{lmcs1119}. 
        
        \begin{definition}
        
            For two \M  strategies $\sigma$, $\sigma'$ and a subset of arcs $A$, we call $\sigma'|_{A}\sigma$ the non-positional strategy that plays like $\sigma'$ until an arc of $A$ is crossed, and then plays like $\sigma$ until the end of the game. We let $\sigma'|^{0}_{A}\sigma = \sigma$ and for all $i \geq 0$, $\sigma'|^{i+1}_{A}\sigma = \sigma'|_{A}(\sigma'|^{i}_{A}\sigma)$. 
        \end{definition}
        
        When $A$ is clear from the context, we omit it and write $\sigma'|^{i}\sigma$. 
        Strategy $\sigma'|^{i}_{A}\sigma$ is the strategy that plays like $\sigma'$ until $i$ arcs from $A$ have been crossed, and then plays like $\sigma$. Hence, we can relate the strategy $\sigma'|_{A}\sigma$
        to a positional strategy in $G[A,\sigma]$ as shown in the next lemma.
        
        \begin{lemma}
        
            \label{EqTransfo}
        
            For two \M  strategies $\sigma$, $\sigma'$ and a subset of arcs $A$, we have: $v^{G}_{\sigma'|_A \sigma} = v^{G[A,\sigma]}_{\sigma'}$
        
        \end{lemma}
        
        \begin{proof}
    In $G$, after crossing an arc from $A$, by definition of $\sigma'|_A \sigma$, \M plays according to $\sigma$. The game being memoryless, from this point, the best response for \m  is to play like $\tau(\sigma) \in BR(\sigma)$. Thus, there is a best response to $\sigma'|\sigma$ of the form $\tau'|\tau(\sigma)$ with $\tau'$ a \m strategy not necessarily positional. Let us consider a play following $(\sigma'|\sigma, \tau|\tau(\sigma))$ with $\tau$ any \m strategy. If the play does not cross an arc of $A$, then there is no difference between this play and a play following $(\sigma',\tau)$ in $G[A,\sigma]$. If an arc of $A$ is used, then by Lemma~\ref{NonChangement} there is no difference between stopping with the value of $G[A,\sigma]$ or continuing in $G$ while following $(\sigma,\tau)$. Thus we have: $v^{G}_{\sigma'|\sigma,\tau|\tau(\sigma)} = v^{G[A,\sigma]}_{\sigma',\tau}.$
            
            Thus, if $\tau'$ is a best response to $\sigma'$ in $G[A,\sigma]$, then $\tau'|\tau(\sigma)$ is a best response to $\sigma'|\sigma$ in $G$. This implies that $v^{G}_{\sigma'|\sigma} = v^{G[A,\sigma]}_{\sigma'}.$
    \end{proof}
        
         We now prove the fact that increasing the values of sinks can only increase the value of the game (a similar lemma is proved in~\cite{auger2014finding}).
    
        \begin{lemma}
        
            \label{IncreaseSink}
            Let $G$ and $G'$ be two identical SSGs except the values of theirs sinks $s \in V_S$,
            denoted respectively by $\val(s)$ and $\val\,'(s)$. If for every $s \in V_S$, $\val\,'(s) \geq \val(s)$, then for every \M  strategy $\sigma$ we have $v^{G'}_{\sigma} \geq v^{G}_{\sigma}$.
        
        \end{lemma}
        
        \begin{proof}
    
    For $s \in V_S$, let $\mathds{P}_{\sigma,\tau}^{x}(\rightarrow s)$ 
            be the probability that the play ends in sink $s$ while starting from vertex $x$, following strategies $(\sigma,\tau)$. For any vertex $x$ we have:
            \[v^{G'}_{\sigma,\tau}(x) = \sum\limits_{s \in V'_S}\mathds{P}_{\sigma,\tau}^x(\rightarrow s)\val'(s) \geq \sum\limits_{s \in V'_S}\mathds{P}_{\sigma,\tau}^x( \rightarrow s)\val(s) = v^{G}_{\sigma,\tau}(x) \]
            This is true for any \m strategy $\tau$, thus $v^{G'}_{\sigma} \geq v^{G}_{\sigma}$.
    \end{proof}
        
        The following proposition is the core idea of GSIA: a strategy which improves on $\sigma$ in the transformed game also improves on $\sigma$ in the original game. 
         The proof relies on a precise analysis of the set of vertices which cannot reach a sink, to deal with the fact that the game is not stopping. We prove that, if $\sigma' \underset{G[A,\sigma]}{\succ} \sigma$ the limit of $v^{G}_{\sigma'|^{i}\sigma}$ is $v^{G}_{\sigma'}$ and the two previous lemmas imply $\sigma'|^{i}\sigma \geq \sigma'|^{i-1}\sigma > \sigma$, which yields the following proposition.
         
        \begin{proposition}\label{proposition:increasing_transfer}
        
            Let $G$ be an SSG, $A$ a subset of arcs of $G$ and $\sigma, \sigma'$ two \M strategies. If $\sigma' \underset{G[A,\sigma]}{\succ} \sigma$ then $\sigma' \underset{G}{>} \sigma$.
        
        \end{proposition}
        
        In order to avoid requiring the game to be stopping, it is necessary to pay particular attention to the set of vertices where the play can loop infinitely and yield value zero, which is a subset of the set of vertices of value $0$. We now prove that a step of GSIA can only reduce this set, which is then used to prove Proposition~\ref{proposition:increasing_transfer}.

         \begin{definition}
        
            For an SSG $G$ and two strategies $(\sigma, \tau)$, an absorbing set $Z$ is a subset of $V \smallsetminus V_S$ such that starting from any vertex of $Z$ and playing according to $(\sigma, \tau)$, there is a probability zero of reaching a vertex of $V \smallsetminus Z$.
        
        \end{definition}
        
        For $\sigma$ and $\tau$ two strategies, $Z(\sigma,\tau)$ is the set of all vertices in some absorbing set under $(\sigma,\tau)$. Hence, $Z(\sigma,\tau)$ is also an absorbing set.
        By definition, a play remains stuck in an absorbing set and can never reach a sink, hence all vertices of an absorbing set have \emph{value zero} under $(\sigma, \tau)$. The next lemma proves the existence of the inclusion-wise maximum over $\tau$ of $Z(\sigma,\tau)$ that we denote by $Z(\sigma)$. An example is given Fig.~\ref{fig:Abs}.
        
        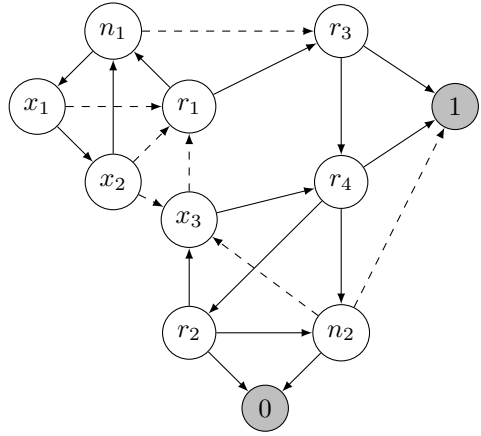
\begin{figure}
\centering

    \begin{tikzpicture}
    
        \node[draw, circle] (x1) at (0,5.5) {$x_1$};
		\node[draw, circle] (n1) at (1,6.5) {$n_1$};
		\node[draw, circle] (x2) at (1,4.5) {$x_2$};
		\node[draw, circle] (r1) at (2,5.5) {$r_1$};
		\node[draw, circle] (r2) at (2,2.5) {$r_2$};
		\node[draw, circle] (x3) at (2,4) {$x_3$};
		\node[draw, circle, fill=gray!50] (0) at (3,1.5) {$0$};
		\node[draw, circle] (r3) at (4,6.5) {$r_3$};
		\node[draw, circle] (r4) at (4,4.5) {$r_4$};
		\node[draw, circle] (n2) at (4,2.5) {$n_2$};
		\node[draw, circle, fill=gray!50] (1) at (5.5,5.5) {$1$};

		\draw[->,>=latex] (n1)--(x1);
		\draw[->,>=latex, dashed] (n1)--(r3);
		\draw[->,>=latex] (x1)--(x2);
		\draw[->,>=latex,dashed] (x1)--(r1);
		\draw[->,>=latex,dashed] (x2)--(r1);
		\draw[->,>=latex,dashed] (x2)--(x3);
		\draw[->,>=latex] (x2)--(n1);
		\draw[->,>=latex] (r1)--(n1);
		\draw[->,>=latex,dashed] (x3)--(r1);
		\draw[->,>=latex] (x3)--(r4);
		\draw[->,>=latex] (r2)--(x3);
		\draw[->,>=latex] (r2)--(n2);
		\draw[->,>=latex] (r2)--(0);
		\draw[->,>=latex] (r1)--(r3);
		\draw[->,>=latex] (r3)--(r4);
		\draw[->,>=latex] (r4)--(r2);
		\draw[->,>=latex] (n2)--(0);
		\draw[->,>=latex,dashed] (n2)--(1);
		\draw[->,>=latex, dashed] (n2)--(x3);
		\draw[->,>=latex] (r3)--(1);
		\draw[->,>=latex] (r4)--(1);
		\draw[->,>=latex] (r4)--(n2);

    \end{tikzpicture}

\caption{Example of an SSG where the $x$, $n$ and $r$ vertices are respectively from $V_{\M}$, $V_{\m}$ and $V_{R}$. The pair of strategy $(\sigma,\tau)$ is displayed as plain arrows. Here $Z(\sigma,\tau) = Z(\sigma) = \{n_1,\;x_1,\;x_2\}$.}
\label{fig:Abs}
\end{figure}
        
        \begin{lemma}\label{lemma:minimal_response_absorbing}
        
            For every \M strategy $\sigma$, there is $\tau \in BR(\sigma)$ such that for every $\m$ strategy $\tau'$, we have $Z(\sigma,\tau') \subseteq Z(\sigma,\tau)$.
        
        \end{lemma}
        
        \begin{proof}
        
            For $\tau$ in $BR(\sigma)$ and $\tau'$ such that $Z(\sigma,\tau') \nsubseteq Z(\sigma,\tau)$, then we define $\tilde{\tau}$ as $\tilde{\tau}(x) = \tau'(x)$ for $x$ in $Z(\sigma,\tau')$ and $\tilde{\tau}(x) = \tau(x)$ otherwise.
            We now prove that $\tilde{\tau} \in BR(\sigma)$ and $ Z(\sigma,\tilde{\tau}) \supseteq Z(\sigma,\tau) \cup Z(\sigma,\tau')$.
            
            Since $\tau$ is a best response to $\sigma$, we have $v_{\sigma,\tau}(x) \leq v_{\sigma,\tau'}(x)$. Moreover, for $x \in Z(\sigma,\tau')$,  $v_{\sigma,\tau'}(x) = 0$
            thus $v_{\sigma,\tau}(x) = 0$. From this,
            we deduce that the two systems
            of linear equations given 
            by  Lemma~\ref{lemma:NSOptConditions}, characterising 
            respectively vectors
            $v_{\sigma,\tau}$ and $v_{\sigma,\tilde{\tau}}$, 
            are exactly the same: for the only vertices
            where $\tilde{\tau}(x)$ and $\tau(x)$ differ
            satisfy $v_{\sigma,\tau}(\tau(x)) = v_{\sigma,\tau}(\tilde{\tau}(x)) = 0$.
            Hence, we have $v_{\sigma,\tau} = v_{\sigma,\tilde{\tau}}$ and $\tilde{\tau} \in BR(\sigma)$.

            For any play under strategies $(\sigma,\tilde{\tau})$ starting in $x \in Z(\sigma,\tau')$, the \m vertices of the play are all in $Z(\sigma,\tau')$ because $\tilde{\tau}$ plays as $\tau'$ on these vertices. Thus, we have $Z(\sigma,\tau') \subseteq Z(\sigma,\tilde{\tau})$.
            For a play starting in $x \in Z(\sigma,\tau)$, either the play reaches a vertex of $Z(\sigma,\tau')$ and then stays in $Z(\sigma,\tau')$ or it plays like $\tau$ and stays in $Z(\sigma,\tau)$.
            Hence, we have $Z(\sigma,\tau) \subseteq Z(\sigma,\tilde{\tau})$.
        \end{proof}

        From this we deduce the following result on the improvement step for GSIA (where absorbing sets are understood in $G$):

        \begin{proposition}\label{proposition:decreasing_absorbing_set}
            Let $G$ be an SSG, $A$ a set of arcs of $G$,  $\sigma$ and $\sigma'$ two \M strategies such that $\sigma' \underset{G[A,\sigma]}{\succ} \sigma$, then $Z(\sigma') \subseteq Z(\sigma)$.
        \end{proposition}
        
        \begin{proof}
        
            Suppose that $Z(\sigma')$ is not a subset of $Z(\sigma)$. From Lemma~\ref{lemma:minimal_response_absorbing} there is $\tau \in BR(\sigma)$ such that $Z(\sigma,\tau) = Z(\sigma)$ and $\tau' \in BR(\sigma')$ such that $Z(\sigma',\tau') = Z(\sigma')$. We write $Z=Z(\sigma')$
            
            Let $X$ be the set of \M  vertices $x$ in $Z(\sigma') \smallsetminus Z(\sigma)$ such that $\sigma(x) \neq \sigma'(x)$; it is nonempty otherwise $Z(\sigma')$ would be an absorbing set for $(\sigma,\tau')$. If $x$ is in $X$, since $\sigma' \underset{G[A,\sigma]}{\succ} \sigma$, we have
            \[v^{G[A,\sigma]}_{\sigma'}(x) > v^{G[A,\sigma]}_{\sigma}(x) \geq 0\]
            Thus, a sink is reached in $G[A,\sigma]$ starting from $x$ under the strategies $(\sigma',\tau')$. Since $Z$ is an absorbing set in $G$ under the same strategies, it implies that all the accessible sinks in $G[A,\sigma]$ are $A$-sinks. Hence, there is at least one arc $e =(y,z) \in A$ with both ends in $Z$ and such that $v_{\sigma}(z) > 0$. We define the vertex $s$ of $Z$ as:
            \[s = \argmax\limits_{z \in Z} \{v_{\sigma}(z)\;|\; \exists y \in V, (y,z) \in A\}\]
            and we let $v = v_{\sigma}(s)$. The value of each vertex in $Z$ is bounded by $v$. Similarly than for $x$, in $G$ under strategies $(\sigma,\tau)$ the value of $s$ is bounded by the value of the vertex leaving $Z$. Such vertices exist since $Z$ is not a subset of $Z(\sigma)$. We now want to show that those vertices all have value strictly lesser than $v$, thus proving a contradiction.
            
            First, since $Z$ is an absorbing set for $(\sigma',\tau')$, all arcs
            leaving a random vertex in $Z(\sigma')$ remain in $Z(\sigma')$ in $G$; 
            this is not dependent on the strategies considered.
            
            Let $E_X \subseteq X$ the set of \M  vertices $x$ of $X$ such that $\sigma(x) \notin Z$ and let $E_N \subseteq Z \cap V_{\m}$ the set of \m  vertices $x$ of $Z$ such that $\tau(x) \notin Z$.
            
            On the one hand, for a \m vertex $x \in E_N$:
            \begin{align*}
                v^{G}_{\sigma}(\tau(x)) & \leq v^{G}_{\sigma}(\tau'(x))&& \text{ Since }\tau = \tau(\sigma)\\
                v^{G}_{\sigma}(\tau'(x))& = v^{G[A,\sigma]}_{\sigma}(\tau'(x)) && \\
                v^{G[A,\sigma]}_{\sigma}(\tau'(x)) & \leq v^{G[A,\sigma]}_{\sigma'}(\tau'(x)) && \text{ Since } \sigma' \underset{G[A,\sigma]}{\succ} \sigma \\
                v^{G[A,\sigma]}_{\sigma'}(\tau'(x)) & \leq v && \text{ Since } \tau'(x) \in Z
            \end{align*}
            Thus, $v^{G}_{\sigma}(\tau(x)) \leq v$. In case
            of equality, we have   $v = v^{G}_{\sigma}(\tau'(x)) = v^{G}_{\sigma}(\tau(x))$;
            hence we can replace $\tau$ by $\bar{\tau}(x)$,which is identical to $\tau$ except that $\bar{\tau}(x) = \tau'(x)$. We have $v_{\sigma,\tau} = v_{\sigma,\bar{\tau}}$ and $Z(\sigma,\tau) = Z(\sigma,\bar{\tau})$. Indeed, according to Lemma~\ref{lemma:NSOptConditions} the only situation that could occur would be to violate the condition $(v)$ by creating an absorbing set. However this would contradict the definition of $\tau$. Thus, we can suppose that for any $x$ in $E_N$, $v_{\sigma}(\tau(x)) < v$.
            
            On the other hand, since $\sigma' \underset{G[A,\sigma]}{\succ} \sigma$ we know that for any $x$ in $E_X$:
            \[v^{G}_{\sigma}(\sigma(x)) < v^{G[A,\sigma]}_{\sigma'}(x) \leq v\]
            Now, for any vertex $x$ of $E = E_X \cup E_N$, let $p_x$ be the probability of $x$ being the first vertex of $E$ reached starting from $s$ following strategies $(\sigma,\tau)$. By conditional expectation:
            
            \[v_{\sigma}(s) = \sum\limits_{x \in E_X}p_{x}v_{\sigma}(\sigma(x)) + \sum\limits_{x \in E_N}p_{x}v_{\sigma}(\tau(x))\]
            
            Thus, $v_{\sigma}(s) < v$ which contradicts the definition of $v$,
            and proves that $Z(\sigma') \subseteq Z(\sigma)$.
        \end{proof}

        We now prove Proposition~\ref{proposition:increasing_transfer}.
        
        \begin{proof}
            We introduce a sequence of non-positional strategies $(\sigma_i)_{i \geq 0}$ defined by $\displaystyle \sigma_i = \sigma'|^{i}\sigma$ for $i \geq 1$. 
            By hypothesis $\displaystyle \sigma' \underset{G[A,\sigma]}{\succ} \sigma$, and by Lemma~\ref{EqTransfo} $v^{G}_{\sigma'|_A \sigma} = v^{G[A,\sigma]}_{\sigma'}$, then we have
            
            \[ v^{G}_{\sigma_1} = v^{G}_{\sigma' | \sigma} = 
            v^{G[A,\sigma]}_{\sigma'} > v^{G[A,\sigma]}_{\sigma} = v^G_\sigma. \]

            
            Hence, by definition, sinks of $G[A,\sigma_1]$ will have at least
            the values of the corresponding sinks in $G[A,\sigma]$. Applying Lemma~\ref{IncreaseSink}, we obtain that $v^{G[A,\sigma_1]}_{\sigma'} \geq v^{G[A,\sigma]}_{\sigma'}$,
            which can also be written as $v^{G}_{\sigma_2} \geq v^{G}_{\sigma_1}$. More generally, we have:
            
            \[\forall i \geq 1,  v^{G}_{\sigma_{i+1}} \geq v^{G}_{\sigma_{i}} \geq v^{G}_{\sigma_1} > v^{G}_{\sigma}.\]
            
            We now prove that $v^{G}_{\sigma'} \geq v^{G}_{\sigma_1}$ to conclude the proof.
            
            From now on, we only consider the game $G$. Fix a vertex $x$ and a \m  strategy $\tau \in BR(\sigma')$ such that $Z(\sigma') = Z(\sigma',\tau)$. From Proposition~\ref{proposition:decreasing_absorbing_set} we know that, $Z(\sigma') \subseteq Z(\sigma)$. It implies that for every $z \in Z(\sigma')$, $v^{G}_{\sigma}(z) = v^{G}_{\sigma'}(z) = 0$ which implies that $v^{G[A,\sigma]}_{\sigma'}(z) = 0$. Thus, $\sigma'(z) = \sigma(z)$. It implies that $Z(\sigma') \subseteq Z(\sigma,\tau)$.
            
            We now only consider $G'$ the game $G$ where we replace every vertex in $Z(\sigma',\tau)$ by a sink of value $0$. Lemma~\ref{NonChangement} directly implies that $v^{G}_{\sigma} = v^{G'}_{\sigma'}$ and $v^{G}_{\sigma'} = v^{G'}_{\sigma'}$. Moreover, when playing following $\sigma_i$ when a vertex of $Z(\sigma')$ is reached, for all possible history, the play will stay in the absorbing set. Thus, $v^{G}_{\sigma_i} = v^{G'}_{\sigma_i}$.
            
            Recall that $\mathds{P}_{\sigma',\tau}^x(\rightarrow s)$ is the probability to reach a sink $s$ in $G'$ while starting in $x$ and following $(\sigma',\tau)$. Let $T^{\sigma',\tau}$ be a random variable defined as the time at which a sink is reached. Note that $T^{\sigma',\tau}$ may be equal to $+\infty$.
            
            For every $i \geq 1$, we use Bayes rule to express the value of $v_{\sigma',\tau}(x)$
            while conditioning on finishing the game before $i$ steps.
            
            \begin{align*}
		    	v_{\sigma',\tau}(x)  =& \mathds{P}(T^{\sigma', \tau} < i)\sum \limits_{s \in V_S} \mathds{P}_{\sigma',\tau}^x(\rightarrow s \mid T^{\sigma', \tau} < i)\val(s) \\& + \mathds{P}(i \leq T^{\sigma', \tau} < +\infty)\sum \limits_{s \in V_S} \mathds{P}_{\sigma',\tau}^x(\rightarrow s \mid +\infty > T^{\sigma', \tau} \geq i)\val(s)
			\end{align*}
			
			If $T^{\sigma_i, \tau} < i$, only $i$ arcs have been crossed, thus at most $i$ arcs from $A$ have been crossed when the sink is reached. Hence $\sigma_i$ acts like $\sigma'$ during the whole play, which yields:
			
			\begin{align*}
		    	v_{\sigma',\tau}(x)  =& \mathds{P}(T^{\sigma_i, \tau} < i)\sum \limits_{s \in V_S} \mathds{P}_{\sigma_i,\tau}^x(\rightarrow s \mid T^{\sigma_i, \tau} < i)\val(s) \\& + \mathds{P}(i \leq T^{\sigma', \tau} < +\infty)\sum \limits_{s \in V_S} \mathds{P}_{\sigma',\tau}^x(\rightarrow s \mid +\infty > T^{\sigma', \tau} \geq i)\val(s)
			\end{align*}
			
			We use Bayes rule in the same way for $v_{\sigma_i,\tau}(x)$
			
			\begin{align*}
		    	v_{\sigma_i,\tau}(x)  =& \mathds{P}(T^{\sigma_i, \tau} < i)\sum \limits_{s \in V_S} \mathds{P}_{\sigma_i,\tau}^x(\rightarrow s \mid T^{\sigma_i, \tau} < i)\val(s) \\& + \mathds{P}(i \leq T^{\sigma_i, \tau} < +\infty)\sum \limits_{s \in V_S} \mathds{P}_{\sigma_i,\tau}^x(\rightarrow s \mid T^{\sigma_i, \tau} \geq i)\val(s)
			\end{align*}
			
		    Since every absorbing vertex in $G$ associated with $\sigma'$ has been turned into a sink, in $G'$ $\mathds{P}(T^{\sigma', \tau} < i) = \mathds{P}(T^{\sigma_i, \tau} < i)$ converges to $1$ when $i$ grows. Hence,
		    both $\mathds{P}(i \leq T^{\sigma', \tau} < +\infty)$ and $\mathds{P}(i \leq T^{\sigma_i, \tau} < +\infty)$ go to $0$ and
		
		    \[\lim\limits_{i \rightarrow +\infty} |v_{\sigma',\tau}(x) - v_{\sigma_i,\tau}(x)| = 0.\]
	
		    Hence, if there was $x$ such that $v_{\sigma'}(x) < v_{\sigma_1}(x)$, we denote $\epsilon = v_{\sigma_1}(x) - v_{\sigma'}(x)$. For some rank $I$ for all $i \geq I$ we have $|v_{\sigma',\tau} - v_{\sigma_i,\tau}| < \epsilon/2$. Which implies $v_{\sigma_i,\tau}(x) < v_{\sigma_1}(x)$. We recall that $v_{\sigma_1}(x) \leq v_{\sigma_i}(x)$. This means that $v_{\sigma_i,\tau} < v_{\sigma_i}(x)$, which contradicts the notion of optimal response against $\sigma_i$. Therefore, we have shown that $\sigma' \underset{G}{\geq} \sigma_1 \underset{G}{>} \sigma.$
        \end{proof}

        As a consequence of all previous lemmas, we obtain the correction of GSIA.
        
        \begin{theorem}\label{thm:MA_terminate}
          GSIA terminates and returns a pair of optimal strategies.
        \end{theorem}
        
        \begin{proof}
        
            We denote by $\sigma_i$ the \M strategy $\sigma$ at the end of the $i$-th loop in Algorithm~\ref{MA}. By induction, we prove that the sequence $\sigma_i$ is of increasing value. Indeed, Line $5$ of Algorithm~\ref{MA} guarantees that $\sigma' \underset{G[A,\sigma]}{\succ} \sigma$, thus Prop.~\ref{proposition:increasing_transfer} implies that  $\sigma' \underset{G}{>} \sigma$,   that is $\sigma_{i+1} > \sigma_{i}$. 
        
    The strategies produced by the algorithm are positional, hence there is only a finite number of them. Since the sequence is strictly increasing, it stops at some point. The algorithm only stops when Line $5$ of Algorithm~\ref{MA} fails to find $\sigma' \underset{G[A,\sigma]}{\succ} \sigma$. In other words, $\sigma$ is optimal in $G[A,\sigma]$. By Lemma~\ref{lemma:optimality_transfer}, $\sigma$ is also optimal in $G$.
        \end{proof}
        
\section{Complexity of GSIA}\label{sec:complexity}

    We analyse the algorithmic complexity of GSIA, by lower bounding the values
    of the sequence of strategies it produces. We obtain a bound on the number of iterations of GSIA depending on the number of random vertices, rather than on the number of  \M  or \m  vertices. Then, we can derive the complexity of any instance of GSIA, by evaluating the cost of computing $\sigma'$ from $\sigma$ in $G[A,\sigma]$.

    \subsection{Values of q-SSGs}
    
        To prove a complexity bound using the values of a strategy,
        we need to precisely characterise the form of these values. 
        In a $2$-SSG, there is a function $f(r)$ such that, for every pair of positional strategies $(\sigma,\tau)$, there is $t \leq f(r)$, such that for every vertex $x$, there is an integer $p_x$, such that $ v_{\sigma,\tau}(x) = \frac{p_x}{t}$
            
        Condon proved in~\cite{condon1992complexity} that $f(r) \leq 4^{r}$. Then Auger, Coucheney and Strozecki improved this to $f(r) \leq 6^{r/2}$ in~\cite{auger2014finding}. We show that $f(r) = q^{r}$ for $q$-SSGs, which gives the improved bound of $f(r) \leq 2^{r}$ for $2$-SSGs. 
        
        \begin{theorem}\label{th:value}
        
            Let $q \geq 1$ and $G$ a $q$-SSG with $r$ random vertices, then for any pair of strategies $(\sigma,\tau)$ there is $t \leq q^{r}$ such that, for every vertex $x$, there is an integer $s_x$ such that, $v_{\sigma,\tau} = \frac{s_x}{t}$.
            
        \end{theorem}
        
        Proof of Th.~\ref{th:value} relies on the matrix tree theorem applied to a directed multigraph representing the game under a pair of strategies. Let us show that $q^r$ is a tight bound for $f(r)$. Consider a Markov chain (an SSG with no \M  nor \m  vertices) with $r+2$ vertices: two sinks $0$ and $1$ and $r$ random vertices $x_1, \ldots, x_r$. Vertex $x_1$ goes to $1$ with probability $1/q$ and to $0$ with probability $(q-1)/q$. For $r \geq i \geq 2$, $x_i$ goes to $0$ with probability $(q-1)/q$ and to $x_{r-1}$ with probability $1/q$. Then, the value of $x_r$ is $q^{-r}$.
        
        Let us remark that a $q$-SSG can be assumed to have 
        all its probability transition of the form $p/q$. The idea here is to notice that it is possible to loop with a certain probability on the same random vertices.

        \begin{lemma}
            Let $G$ be a $q$-SSG, then there is $G'$ a $q$-SSG with
            the same vertices which defines the same expectation $\esps{x_0}{\sigma, \tau}{\cdot \vert \cdot}$ and such that for all $x \in V_R$ and all $x' \in N^+(x)$ then there is an integer $p_{x,x'}$ such that $p_x(x') = p_{x,x'}/q$.
        \end{lemma}
        
        \begin{proof}
        
            For $a$ a random vertex in $G$, and $q_a < q$ such that for every other vertex $x$ in $G$ there is $p_x \in \mathbb{N}$ and a probability $p_x/q_a$ to go directly from $a$ to $x$, we change those probabilities to $p_x/q$ and we add a probability $p/q$ to stay in $a$, where:
            \[p = q - \sum\limits_{x \in V} p_x.\]
        \end{proof}

        Now, we state the classical matrix-tree theorem that we use in our proof (see e.g. \cite{chaiken1978matrix}). Let $G$ be a directed multigraph with $n$ vertices, then the \emph{Laplacian matrix} of $\mathcal{G}$ is a $n \times n$ matrix $L(\mathcal{G}) = (l_{i,j})_{i,j \leq n}$ defined by:
        
        \begin{romanenumerate}
            \item $l_{i,j}$ equals $-m$ where $m$ is the number of arcs from $i$ to $j$.
            \item $l_{i,i}$ is the number of arcs going to $i$, excluding the self-loops.
        \end{romanenumerate}
        
        \begin{theorem}[Matrix tree theorem for directed multigraphs]
        
            \label{MTT}
        
            For $G = (V,E)$ a directed multigraph with vertices $V = \{v_1, \ldots, v_k\}$ and $L$ its Laplacian matrix, the number of spanning trees rooted at $v_i$ is $\det(\hat{L}_{i,i})$ where $\hat{L}_{i,i}$ is the matrix obtained by deleting the $i$-th row and column from $L$.
        
        \end{theorem}
        
         We can now prove Th.~\ref{th:value}.
         
         \begin{proof}[Proof of Th.~\ref{th:value}]
         
             The beginning of the proof is the same as in~\cite{condon1993algorithms} and~\cite{auger2014finding}. We start by transforming the game with fixed strategies in a Markov Chain with equivalent values. Then, we show that the value of each vertex can be written $\frac{\det B_i}{\det q(I-A)}$ using Cramer rule, for $B_i$ and $A$ two matrix which will be carefully defined. To conclude, we will show that $\det q(I-A) < q^r$ by creating a graph obtain from our initial game and using Th.~\ref{MTT}.
             
             We consider a $q$-SSG $G$ and two positional strategies $\sigma$ and $\tau$. Without loss of generality, we can restrict ourselves to the computation of non-zero, non-sink values. Thus, each vertex has a non-zero probability to reach the $1$-sink. To compute the values $v_{\sigma,\tau}$, we can consider $G_A$ an SSG with vertices $V_R \cup V_S$: the random vertices and the sinks of $V$. The value of the sinks is not changed and the probability distribution $p'_x$ is defined as follows. For $x \in V_R$ and $x'$ in $G_A$, we call $M_{x,x'}$ the set of \M and \m vertex $y$ in $N^{+}(x)$ such that there is a path following only arcs of $\sigma$ and $\tau$ from $y$ to $x'$. We then have \[p'_x(x') = \sum\limits_{y \in M_{x,x'}} p_x(y)\]
             
             The graph $G_A$ has $r+2$ vertices that we denote by $a_1, \ldots, a_{r+1}, a_{r+2}$ where $a_{r+1}$ is the $0$-sink and $a_{r+2}$ is the $1$-sink. Let $b$ be the $r$-dimensional column vector with $\displaystyle b_i = p'_{a_i}(a_{r+2})$. We define $A$ the $r \times r$ matrix, with $A_{i,j} = p'_{a_i}(a_j)$.
        
        	The values of the random vertices are defined by the vector $z$ that satisfies the following equation:
        	\[z = Az + b\]
        	Let $I$ be the identity matrix, $(I- A)$ is invertible because each random vertex has access to a sink and every eigenvalue of $A$ is strictly less than $1$. We refer to~\cite{condon1993algorithms} for details. Hence, the equation has a unique solution and $z$ is also solution of:
        	\[q(I-A)z = qb\]
        	Hence, under the strategies $\sigma,\tau$, the value $z_i$ of a random vertex $a_i$ given by the Cramer rule is
        	\[z_i = \frac{\det B_i}{\det q(I-A)}\]
        	where $B_i$ is the matrix $q(I-A)$ where the $i$-th column has been replaced by $qb$. The value $\det B_i$ is an integer. See~\cite{auger2014finding} for more details. Our goal is now to bound $\det q(I-A)$.
        	
        	From the graph $G_A$, we construct the graph $G'$ by inverting all arcs, and 
        	duplicating an arc of probability $p/q$ into $p$ arcs of probability $1/q$. We also add an arc coming from the $1$-sink to the $0$-sink and one from the $0$-sink towards the $1$-sink. Figure~\ref{fig:MTT} shows an example of the transformation from $G$ to $G'$. The Laplacian $L$ of $G'$ is thus the following matrix.
        	
        	\[L = \left(\begin{array}{c|c} q(I-A)^{T} & B \\ \hline 0 & \begin{matrix} 1&1\\1&1 \end{matrix}\end{array}\right)\]
        	
        	Indeed, every random vertex has indegree $q$ minus the number of loops. Thus the number of spanning trees of $G'$ rooted in the $1$-sink is equal by Th.~\ref{MTT} to $\det \hat{L}_{r+2,r+2}$ where we have
        	
        	\[\hat{L}_{r+2,r+2} = \left(\begin{array}{c|c} q(I-A)^{T} & B' \\ \hline 0 & 1 \end{array}\right).\]
        	
        	In other words, the number of spanning trees of $G'$ is equal to $\det q(I-A)$. Furthermore, each spanning tree contains exactly one incoming arcs for every random vertices, and the arc $(a_{r+2},a_{r+1})$ has to be used. Thus, there is at most $q^{r}$ spanning trees rooted in $G'$ and:
        	\[\det q(I-A) \leq q^{r}.\]
        \end{proof}
        
        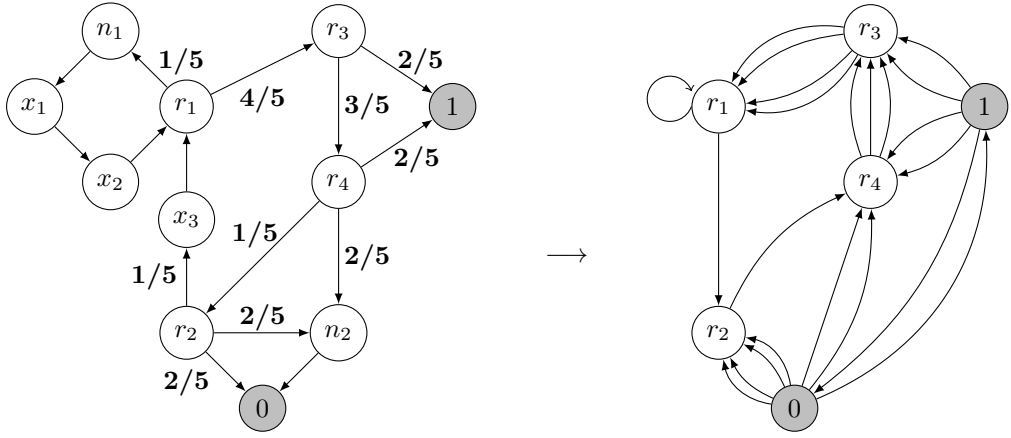
\begin{figure}
\centering

    \begin{tikzpicture}
    
        \node[draw, circle] (x1) at (0,5.5) {$x_1$};
		\node[draw, circle] (n1) at (1,6.5) {$n_1$};
		\node[draw, circle] (x2) at (1,4.5) {$x_2$};
		\node[draw, circle] (r1) at (2,5.5) {$r_1$};
		\node[draw, circle] (r2) at (2,2.5) {$r_2$};
		\node[draw, circle] (x3) at (2,4) {$x_3$};
		\node[draw, circle, fill=gray!50] (0) at (3,1.5) {$0$};
		\node[draw, circle] (r3) at (4,6.5) {$r_3$};
		\node[draw, circle] (r4) at (4,4.5) {$r_4$};
		\node[draw, circle] (n2) at (4,2.5) {$n_2$};
		\node[draw, circle, fill=gray!50] (1) at (5.5,5.5) {$1$};
		
		\node (tr) at (7,3.5) {$\longrightarrow$};
		
		\node[draw, circle] (r1b) at (9,5.5) {$r_1$};
		\node[draw, circle] (r2b) at (9,2.5) {$r_2$};
		\node[draw, circle, fill=gray!50] (0b) at (10,1.5) {$0$};
		\node[draw, circle] (r3b) at (11,6.5) {$r_3$};
		\node[draw, circle] (r4b) at (11,4.5) {$r_4$};
		\node[draw, circle, fill=gray!50] (1b) at (12.5,5.5) {$1$};

		\draw[->,>=latex] (n1)--(x1);
		\draw[->,>=latex] (x1)--(x2);
		\draw[->,>=latex] (x2)--(r1);
		\draw[->,>=latex] (r1)--(n1);
		\draw[->,>=latex] (x3)--(r1);
		\draw[->,>=latex] (r2)--(x3);
		\draw[->,>=latex] (r2)--(n2);
		\draw[->,>=latex] (r2)--(0);
		\draw[->,>=latex] (r1)--(r3);
		\draw[->,>=latex] (r3)--(r4);
		\draw[->,>=latex] (r4)--(r2);
		\draw[->,>=latex] (n2)--(0);
		\draw[->,>=latex] (r3)--(1);
		\draw[->,>=latex] (r4)--(1);
		\draw[->,>=latex] (r4)--(n2);
		
		\draw [->] (r1b.west)arc(-20:-340:0.3);
		\draw[<-,>=latex] (r1b) to[bend right = 20] (r3b);
		\draw[<-,>=latex] (r1b) to[bend right = 35] (r3b);
		\draw[<-,>=latex] (r1b) to[bend left = 20] (r3b);
		\draw[<-,>=latex] (r1b) to[bend left = 35] (r3b);
		\draw[<-,>=latex] (r2b)--(r1b);
		\draw[<-,>=latex] (r2b) to[bend right = 20] (0b);
		\draw[<-,>=latex] (r2b) to[bend right = 35] (0b);
		\draw[<-,>=latex] (r2b) to[bend left = 20] (0b);
		\draw[<-,>=latex] (r2b) to[bend left = 35] (0b);
		\draw[<-,>=latex] (r3b) to[bend right = 20] (1b);
		\draw[<-,>=latex] (r3b) to[bend left = 20] (1b);
		\draw[<-,>=latex] (r3b) to[bend right = 20] (r4b);
		\draw[<-,>=latex] (r3b) to[bend left = 20] (r4b);
		\draw[<-,>=latex] (r3b)--(r4b);
		\draw[<-,>=latex] (r4b) to[bend right = 20] (1b);
		\draw[<-,>=latex] (r4b) to[bend left = 20] (1b);
		\draw[<-,>=latex] (r4b)--(0b);
		\draw[<-,>=latex] (r4b) to[bend left = 20] (0b);
		\draw[<-,>=latex] (r4b) to[bend right = 20] (r2b);
		\draw[<-,>=latex] (0b) to[bend right = 20] (1b);
		\draw[->,>=latex] (0b) to[bend right = 35] (1b);

		\draw (1.5,6.1) node[right]{\bfseries 1/5};
		\draw (3,5.9) node[below]{\bfseries 4/5};
		\draw (2,3.25) node[left]{\bfseries 1/5};
		\draw (r2.south) node{} node[below]{\bfseries 2/5};
		\draw (3,2.4) node[above]{\bfseries 2/5};
		\draw (4.65,6.1) node[right]{\bfseries 2/5};
		\draw (3.95,5.5) node[right]{\bfseries 3/5};
		\draw (4.6,4.8) node[right]{\bfseries 2/5};
		\draw (3.95,3.5) node[right]{\bfseries 2/5};
		\draw (2.9,3.5) node[above]{\bfseries 1/5};

    \end{tikzpicture}

\caption{Example of a transformation of a graph $G$ into a graph $G'$}
\label{fig:MTT}
\end{figure}
    
    \subsection{Bounding the Number of Iterations of GSIA}
    
    GSIA produces a sequence of strictly increasing positional \M  strategies. The number of positional \M strategies is bounded by $\left|\Sigma^{\M}\right| = \prod\limits_{x \in V_{\M}}\deg(x)$, hence the number of iterations of GSIA is bounded by this value. If we consider the case of a binary SSG (all vertices of outdegree $2$), we have the classical bound of $|\Sigma^{\M}|= 2^{n}$ iterations. The best known bound for a deterministic algorithm is $2^{n}/n$ iterations obtained for Hoffman-Karp algorithm~\cite{tripathi2011strategy}, which is not far from the trivial bound of $2^n$ iterations.
    
        We give a bound for $q$-SSG, which depends on $q$ and $r$ the number of random vertices. 
        The difference of two values written as $a/b$ and $c/d$, with $a$ and $b$ less than $q^{-r}$ 
        is more than $q^{-2r}$. Hence, if a value increases in GSIA, it increases at least by $q^{-2r}$.
        Using the classical notion of switch and anti-switch~\cite{tripathi2011strategy}, recalled bellow, we can prove
        that all vertices which have their value increased by a step of GSIA, are increased by at least $q^{-r}$.

        \begin{theorem}\label{th:iteration}
        
            For $G$ a $q$-SSG with $r$ random vertices and $n$ \M  vertices, the number of iterations of GSIA is at most $nq^{r}$.
        
        \end{theorem}
        
        We introduce the notion of switch and anti-switch, to prove that the improvement is at least $q^{-r}$ rather than $q^{-2r}$.
        
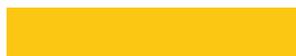
        \begin{definition}
            A switch (resp. an anti-switch) of a \M strategy $\sigma$ 
            with switched set $S~\subseteq~V_\M$ is a strategy $\sigma_S$ defined by $\sigma_S(x) = \sigma(x)$ for $x \notin S$, and
            satisfying 
            $v_{\sigma}(\sigma(x))~<~v_{\sigma}(\sigma_S(x))$ (resp. $v_{\sigma}(\sigma(x))~>~v_{\sigma}(\sigma_S(x))$)
            for  $x \in S$  (hence $\sigma_S(x) \neq \sigma(x)$).

        \end{definition}
        
        A common tool to solve SSGs is the fact that a switch increases the value of a strategy, while an anti-switch decreases it. Within our framework of transformed game, it is extremely simple to prove.

        \begin{lemma}\label{lemma:switch}
        If $\sigma_S$ is a switch of $\sigma$, then $\sigma_S > \sigma$.
        If $\sigma_S$ is an anti-switch of $\sigma$, then $\sigma_S < \sigma$.
        \end{lemma}
        \begin{proof}
        Consider $G[A,\sigma]$ the game obtained from $G$, where $A$ is the set of all arcs of $G$. Let us consider $x$ a vertex switched in $\sigma'$, that is with $v^G_{\sigma}(\sigma(x)) < v^G_{\sigma}(\sigma'(x))$. Then, because all arcs are in $A$, we have $v^{G[A,\sigma]}_{\sigma}(x) = v^G_{\sigma}(\sigma(x))$ and
        $v^{G[A,\sigma]}_{\sigma'}(x) = v^G_{\sigma}(\sigma'(x))$. Hence, $v^{G[A,\sigma]}_{\sigma}(x) < v^{G[A,\sigma]}_{\sigma'}(x)$ and for  $v^G_{\sigma}(\sigma(x)) \geq v^G_{\sigma}(\sigma'(x))$, $\sigma(x)=\sigma'(x)$, 
        which implies $\sigma' \underset{G[A,\sigma]}{\succ} \sigma$. Prop.~\ref{proposition:increasing_transfer} proves $\sigma' \underset{G}{>} \sigma$.
        
        The proof is the same for an anti-switch, since $\sigma \underset{G[A,\sigma]}{\succ} \sigma' \Rightarrow \sigma \underset{G}{>} \sigma'$ (which can be proved similarly as  Prop.~\ref{proposition:increasing_transfer}, while keeping in mind that in the decreasing case, creating absorbing set lowers the value). 
        \end{proof}
        
        We use the previous lemma to prove Th.~\ref{th:iteration}.
        
        \begin{proof}
        Let us consider $\sigma$ the strategy computed at some point by GSIA and $\sigma'$ the next strategy. By Prop.~\ref{proposition:increasing_transfer}, $\sigma < \sigma'$. 
        Hence, by Lemma~\ref{lemma:switch}, $\sigma'$ cannot be an anti-switch of $\sigma$. Thus, there is a \M  vertex $x$ such that $v_{\sigma}(\sigma(x)) < v_{\sigma}(\sigma'(x))$. We recall that $\sigma'(x)$ denotes the successor of $x$ under strategy $\sigma'$. 
        
        Since $\sigma < \sigma'$, we have $v_{\sigma}(x) = v_{\sigma}(\sigma(x)) < v_{\sigma}(\sigma'(x)) \leq v_{\sigma'}(\sigma'(x)) = v_{\sigma'}(x)$.
        We now evaluate  $v_{\sigma}(\sigma'(x)) - v_{\sigma}(\sigma(x))$. 
        In the game $G$, under the strategies $\sigma,\tau(\sigma)$, Th.~\ref{th:value} implies that for some $t \leq q^{r}$,  $v_{\sigma}(\sigma(x)) = p/t$ and 
        $v_{\sigma}(\sigma'(x)) = p'/t$. We have $p/t < p'/t$, thus $p'/t - p/t \geq 1/t \geq 1/q^{r}$. Hence, the value of some \M  vertex increases by $1/q^r$ in each iteration of GSIA. Since there are $n$ \M  vertices and their values are bounded by $1$, there are at most $nq^{r}$ iterations.
        \end{proof}
        
        The complexity of GSIA is the number of iterations given by Th.~\ref{th:iteration}, multiplied by the complexity of an iteration. In an iteration, there are two sources of complexity: constructing the game $G[A,\sigma]$ and finding an improving strategy $\sigma'$ in $G[A,\sigma]$. To construct the game, $v_\sigma$ is computed by solving a linear program of size $m$ up to precision $p=q^r$. Let $C_1(m,p)$ be the complexity of computing $v_\sigma$, then the best bound is currently in $O(m^{\omega}\log(p))$~\cite{jiang2020faster}, with $\omega$ the current best bound on the matrix multiplication exponent. Let $C_2(n,r,q)$ be the complexity of computing $\sigma'$, the complexity of an iteration is in $O(nq^r(C_1(n + r, q^r) + C_2(n,r,q))$. 
        
        We obtain a better complexity, when $C_2(n,r,q) = O(C_1(n, q^r)r/n)$, which is the case for most instances of GSIA mentionned in this article. The number of iterations is only $rq^r$ if we can guarantee that a random vertex increases its value at each step. When no random vertex is improved, the cost of computing $G[A,\sigma]$ can be made smaller, which yields the following theorem.

        \begin{theorem}\label{th:complexity}
            Let $G$ be a $q$-SSG with $r$ random vertices and $n$ \M  vertices. If  $C_2(n,r,q) = O(C_1(n, q^r)r/n)$, then the complexity of GSIA is in $O(rq^{r}C_1(n,q^r))$.
        \end{theorem}
        
        \begin{proof}
         We assume that $r < n$, otherwise the theorem is trivial. 
        Let $\sigma'$ be the strategy computed by GSIA at some point, improving on the strategy $\sigma$. GSIA must compute $G[A,\sigma']$, and thus $v_{\sigma'}$ and we explain a method to do so efficiently.
        
        We assume that the order of the values (in $G$) of the random vertices is the same for $\sigma$ and $\sigma'$. Then, knowing this order and $\sigma'$, it is easy to compute $\tau(\sigma')$ a best response to $\sigma'$ in $O(r\log(r)+n)$ time~\cite{andersson2008deterministic}. Then, we can compute the values $v_{\sigma',\tau(\sigma')}$ in time $O(C_1(r,q^r))$, since it is done by solving a linear system of dimension $r$ with precision $q^r$, a task which is simpler than solving a linear program. Since $C_1(r,q^r)$ is at least quadratic in $r$, then $C_1(r,q^r) < C_1(n,q^r) r/n$ and by hypothesis $C_2(n,r,q) = O(C_1(n, q^r)r/n)$, hence a step is of complexity at most $O(C_1(n,q^r)r/n)$. There are at most $nq^r$ such steps, for a total complexity of $O(rq^rC_1(n,q^r))$.
        
        We need to detect when the assumption that the values of the random vertices are the same for $\sigma$ and $\sigma'$ is false. If $v_{\sigma',\tau(\sigma')}$ satisfies the optimality conditions at the \m vertices, then $\tau(\sigma')$ is a best response. Otherwise, we compute the best response by solving a linear program in time $C(n,q^r)$. In that case, the order of the random vertices has changed:
        there are two vertices $x_1$ and $x_2$ such that $v_{\sigma}(x_1) < v_{\sigma}(x_2)$ and $v_{\sigma'}(x_1) > v_{\sigma'}(x_2)$. Hence, $v_{\sigma'}(x_1) > v_{\sigma}(x_2)$, which implies that $v_{\sigma'}(x_1) - v_{\sigma}(x_1) > v_{\sigma}(x_2) - v_{\sigma}(x_1) > q^{-r}$.
        
        We have proved that when the random order changes, the value of some random vertex increases by at least $q^{-r}$, hence there are at most $rq^r$ such steps. The complexity from these steps is bounded by $O(rq^{r}C_1(n,q^r))$, which proves the theorem.
        \end{proof}

\section{Two Instances of GSIA}\label{sec:GSIinstances}

    As previously mentioned, all known strategy improvement algorithms can be viewed as particular instances of GSIA. This includes e.g. switch-based algorithms, like Hoffman-Karp algorithm \cite{condon1993algorithms, tripathi2011strategy} or Ludwig's recursive algorithm~\cite{ludwig1995subexponential}. With the help of GSIA it also becomes very easy to derive new algorithms, by transforming the game into polynomial time solvable instances, such as almost acyclic games~\cite{auger2014finding}. We detail all these old and new algorithms in Section~\ref{section:salad}.
    
    In this section, we focus instead on two particular instances (or family of instances) of GSIA, for which we obtain new complexity bounds using the results of the previous sections.

    \subsection{GSIA and f-strategies} \label{sec:gh}

    The strategy improvement algorithm proposed by Gimbert and Horn in~\cite{gimbert2008simple} (denoted by GHA) can be viewed as an instance of GSIA where the set $A$ of fixed arcs is the set $R$ of all arcs going out of random vertices, and the improvement step in the subgame $G[R,\sigma]$ consists in taking an optimal strategy.
    In this case, the subgame $G[R,\sigma]$ is deterministic (random vertices are connected to sinks only and can be replaced by sinks), hence  optimal values in $G[R,\sigma]$ depend only on the relative ordering of the values $v_\sigma(x)$ for sink and random vertices $x$ of $G$.
    These values can be computed in $O(r\log(r)+n)$ time~\cite{andersson2008deterministic}. In the original paper \cite{gimbert2008simple}, the algorithm is proposed in a context where the number
    of sinks is two, but we generalise their definitions to our context.
    
    Consider a total ordering $f$ on $V_R \cup V_S$, $f : x_1 < x_2 < \cdots < x_{r+s}$,
    where $s$ is the number of sinks. An
    $f$-strategy corresponding to this ordering is an optimal $\M$ strategy in the game where
    the $s+r$ vertices above are replaced by sinks with new values satisfying $\val(x_1) < \val(x_2) < \cdots < \val(x_{r+s}).$ Clearly, this strategy does not depend on the actual values given but only on $f$. Note that if several $f$-strategies exist for a given $f$,
    they share the same values on all vertices.
    
    Algorithm GHA produces an improving sequence of $f$-strategies, and the two sinks of value zero and one are always first and last in the order, hence its number of iterations is bounded by $r!$, the total number of possible orderings of the random vertices. We extend this result to a large class of instances of GSIA: let us call Optimal-GSIA (Opt-GSIA), the meta algorithm obtained from Algorithm~\ref{MA} with two additional constraints:
    \begin{itemize}
    \item the set $A$ of fixed arcs is the same at each step of Algorithm~\ref{MA};
    \item at line $5$, the improving strategy $\sigma'$ is the optimal strategy in $G[A,\sigma]$.
    \end{itemize}
    
    All classical algorithms captured by GSIA, or new ones presented in this article are in fact instances of Opt-GSIA. We now show that Opt-GSIA has an iteration number similar to GHA. Since we have proved a bound of $nq^r$ iterations, by Th.~\ref{th:iteration}, Opt-GSIA has essentially the best known number of iterations, for $q$ small and large (the latter being interesting in the case of random vertices with large degree and arbitrary probability distributions).

       \begin{theorem}\label{theorem:fstrat}
        
       Consider an SSG $G$ and a set of arcs $A$ containing $k$ arcs out of \M or \m vertices.
       Then Algorithm Opt-GSIA  runs in at most $\min((r+k)q^{r},(r+k)!)$ iterations.
        
       \end{theorem}
       
       \begin{proof}
        Let $\sigma$ be one of the iterated $\M$ strategies obtained by an instance of Opt-GSIA,
        and $\sigma'$ be an optimal strategy in $G[A,\sigma]$. Then $\sigma'$ is 
        consequently an $f$-strategy in $G[A,\sigma]$, where $f$ is the ordering on $V_R \cup V_A$ (where $V_A$ is the set of $A$-sinks) which is induced by the value vector $v^{G[A,\sigma]}_{\sigma'}$ (if vertices have the same value, just arbitrarily decide
        their relative ordering in $f$).
        
        Since strategies produced by the algorithm strictly increase in values by Prop.~\ref{proposition:increasing_transfer}, they must be all distinct. Hence, the order
        $f$ must be distinct at each step of the algorithm, which proves that Opt-GSIA does at most $(r+k)!$ iterations.
        
        Moreover, at every step the value in $G$ of at least one vertex in $V_R \cup V_A$ must improve, by at least $q^{-r}$ because of Th.~\ref{th:value}. Since the value of these vertices is bounded by $1$, the number of iterations of Opt-GSIA is bounded by $(r+k)q^{r}$.
       \end{proof}

    \subsection{Generalised Gimbert and Horn's Algorithm}
    
    Th.~\ref{theorem:fstrat} gives a competitive bound on the number of iterations for a strategy improvement algorithm, but the algorithmic complexity of an instance of Opt-GSIA also depends on how we find an optimal solution in $G[A,\sigma]$. We now present a class of instances of Opt-GSIA generalising GHA, with two interesting properties:  there is a simple condition on $A$, which guarantees that $G[A,\sigma]$
    is solvable in polynomial time, and they can be precisely compared to Ibsen-Jensen and Miltersen's algorithm (denoted by IJMA)~\cite{ibsen2012solving}, which is the current best deterministic algorithm for $2$-SSGs.
    
    Let us describe IJMA, restated in our framework, and generalised to $q$-SSGs. 
    IJMA is not a strategy improvement algorithm but a {\it value iteration algorithm},  which keeps a vector of values for random vertices, here denoted by $v^{\text{IJMA}}_i$ at step $i$.  This vector is updated in the following way:
    \begin{itemize}
    \item first, an optimal $f$-strategy is computed in the deterministic game $G[R,v^{\text{IJMA}}_i]$, and we denote the values of this game by $v'^{\text{IJMA}}_i$  (remember that $R$ is  the set of arcs going out of random vertices in $G$) ;
    \item second, the vector $v^{\text{IJMA}}_i$ is updated on every random vertex by $$v^{\text{IJMA}}_{i+1}(x) = \sum\limits_{y \in N^+(x)} p_x(y) \cdot v'^{\text{IJMA}}_i(y).$$
    \end{itemize}
    As can be seen, IJMA has an almost linear update complexity
    since $G[R,\sigma]$ is a deterministic game and can be solved in $O(r\log(r)+n)$ time.
    
    As noted in Sec.~\ref{sec:gh}, recall that Optimal-GSIA with $R$ as a fixed set of arcs is equivalent to Gimbert and Horn's algorithm (GHA).
    When $A \subseteq R$ in Opt-GSIA, we obtain a generalisation of GHA,
    which can be precisely compared to IJMA as shown in the next theorem.

    \begin{theorem}\label{theorem:comparison}
    Opt-GSIA, with $A \subseteq R$ needs less iterations than IJMA to find the optimal values on any input.
    \end{theorem}
    
    \begin{proof}
        We denote by $\sigma_i$ the strategy obtained after $i$ steps of an instance of Opt-GSIA, where $A \subseteq R$.
        We prove by induction on $i$ that $v_{\sigma_i} \geq v^{\text{IJMA}}_i$ (on random vertices).
        In IJMA, the value vector is initialised to $0$ at the first step, hence any choice of initial strategy for Opt-GSIA guarantees a larger value on random vertices and satisfies the induction hypothesis.
        
        Now assume that $v_{\sigma_i} \geq v^{\text{IJMA}}_i$ for some $i$.
        First, we have
        \begin{equation}
        v_{\sigma_{i+1}} = v^{G[A,\sigma_{i+1}]}_{\sigma_{i+1}}
        \end{equation}
        by Lemma~\ref{NonChangement} and 
        \begin{equation}
        v^{G[A,\sigma_{i+1}]}_{\sigma_{i+1}} \geq v^{G[A,\sigma_i]}_{\sigma_{i+1}} 
        \end{equation}
        since $\sigma_{i+1} < \sigma_i$ by Proposition~\ref{proposition:increasing_transfer}. Now
        $v^{G[A,\sigma_i]}_{\sigma_{i+1}}$ is, by definition of Opt-GSIA, an optimal value vector of $G[A,\sigma_i]$, but 
        optimal values are larger in $G[A,\sigma_i]$  than in $G[R,\sigma_i]$
        since $A \subseteq R$, and on the other hand, optimal values of $G[R,\sigma_i]$ are larger than those of 
        $G[R,v^{\text{IJMA}}_i]$, using Lemma~\ref{IncreaseSink} and the induction hypothesis,
        the latter being $v'^{\text{IJMA}}$ by definition of IJMA. Putting these together, 
        we have proved that 
        \begin{equation}
        v_{\sigma_{i+1}} \geq v'_{\text{IJMA}}.
        \end{equation}

        Consider now a random vertex $x$. Considering
        the optimality conditions of Lemma~\ref{lemma:NSOptConditions}
        for $v_{\sigma_{i+1}}$ and the definition of $v^{\text{IJMA}}_{i+1}(x)$, 
        we see that
        \begin{align*}
        v_{\sigma_{i+1}}(x) & = \sum\limits_{y \in N^+(x)} p_x(y) v_{\sigma_{i+1}}(y)
                            &  \geq \sum\limits_{y \in N^+(x)} p_x(y) v'^{\text{IJMA}}_{i+1}(y)
                            & = v^{\text{IJMA}}_{i+1}
        \end{align*}

        This concludes the induction and the result follows.
        \end{proof}
    
    We have proved that on every game, instances of Opt-GSIA with $A\subseteq R$ make less iteration than IJMA; and it can need dramatically less of them. Indeed, the analysis of IJMA~\cite{ibsen2012solving} relies on finding an extremal input for the algorithm, which happens to have no \M nor \m vertices. This extremal input is solved in one iteration of Opt-GSIA with the help of the "best response" step, and Opt-GSIA is then exponentially faster. We have yet no result to quantify how faster Opt-GSIA is in the general case, but we suspect the number of iterations of Opt-GSIA to be much smaller in many cases.
    
    While the number of iterations of Opt-GSIA is better than the number of iterations of IJMA, one should take into account the complexity of a single iteration.
    In general, there is no algorithm for solving an iteration of Opt-GSIA
    in polynomial time, since when 
    $A = \emptyset$ it is equivalent to solving any SSG; but
    let us consider a mild condition ensuring that it will be the case. 
    Suppose that $A$ contains at least one arc with transition
    probability at least $1/n$ out of every random vertex of
    $G$. This is not a very restrictive property, since there is a least one such arc for each random vertex, and thus at least $2^r$ sets $A$ have this property. For $q$-SSGs, with $q$ fixed, the condition can be simplified into saying that $A$ contains at least an arc out of each random vertex.
    In this case, in the subgame $G[A,\sigma]$, there is a probability of stopping on a sink of a least $1/n$ when going through a random vertex. Hence, for this game, the value iteration algorithm
    (as presented in~\cite{ibsen2012solving}), converges in polynomial time in $n$ to a value vector which is close enough to the optimal value vector so that we can recover it. 
    Then, at the end of an iteration of Opt-GSIA, a best response must be computed,  in time $O(rn^{\omega})$. This should be compared to IJMA, whose iterations only requires an almost linear time.

    To remedy this, we propose an hybrid version between Opt-GSIA and IJMA that
    combines the good properties of both algorithms:
    do the same value iteration algorithm as IJMA, but once every 
    $rn^{\omega-1}$ iterations, compute a best response, in time $O\left(rn^{\omega-1}\right)$, to update the values as in Opt-GSIA rather than doing a value propagation. This hybrid version enjoys the same complexity as IJMA since the overhead from the best response computation is in constant time per iteration. Moreover, the proof of Th.~\ref{theorem:comparison} shows that it needs less iterations than
    IJMA, and exponentially so for the extremal input which is a cycle of random vertices.
    
    \subsection{Condon's Converge From Below Algorithm}
    
    In~\cite{condon1993algorithms}, Condon  first presents a faulty algorithm
    (the Naive Converge From Below Algorithm) and then a correct modified version, the Converge From Below (CFB) Algorithm. This algorithm proceeds by improving a value vector iteratively, but we show here that is in fact
    a disguised strategy improvement algorithm, that can be seen as an instance of
    Opt-GSIA. This gives us a proof of convergence of the CFB algorithm
    in the general, non-stopping case (whereas Condon has the assumption that the game is stopping in her proof), and also
    bounds on the number of iterations (none are given in the original paper) by Theorem \ref{th:iteration} and Theorem \ref{theorem:fstrat}.
    
    The CFB algorithm is restated with some clarifications on listing \ref{alg:CVB}
    (we omit the details of the linear program, see~\cite{condon1993algorithms}).
    The algorithm uses two properties of a vector, that we now define. First,
    vector $v$ is {\it feasible} if 
    \begin{romanenumerate}
                \item For $s \in V_S$, $v(s) = \val(s)$
                \item For $r \in V_R$, $v(r) = \sum\limits_{x \in N_r} p_r(x)v(r)$
                \item For $x \in V_{\m}$, $v(x)  \leq \min\limits_{y \in N^{+}(x)} v(y)$
                \item For $x \in V_{\M}$, $v(x) \geq  \max\limits_{y \in N^{+}(x)} v(y)$.
    \end{romanenumerate}
    A feasible vector is {\it stable} at $x$ a $\m$ vertex (resp. $\M$ vertex) if satisfies condition $(iii)$ (resp. condition $(iv)$) of feasibility for $x$ with an equality. 
    
    We now show by induction that the CFB algorithm is
    equivalent to the instance of Opt-GSIA where all $\m$ vertices are fixed, i.e. $A$ is the set of arcs entering \m vertices.
    Let $A_\m$ denote this set.
    
    To see this, suppose that at the beginning of line 5 of CFB, 
    Vector $v_r$ is the value vector of a $\M$-strategy $\sigma$ in $G$. Then:
    
    \begin{itemize}
    \item at Line 5, we ``update $v$ as the feasible vector where all $\m$ vertices $x$ have value $v_r(x)$ and all $\M$ vertices are stable’’. This amounts to finding
    a $\M$-strategy $\sigma'$ which satisfies optimality conditions
    in $G[A_\m, \sigma]$, i.e. an 
    optimal strategy for $\M$ in this subgame. This
    is exactly the subgame improvement step of
    Opt-GSIA. At 
    the end of this step, $v$ is the optimal value vector
    in $G[A_\m, \sigma]$ ;
    
    \item in the next loop, at Line 4 of CFB, we ``compute the value vector $v_r$ of an optimal response  to the $\M$ strategy that plays greedily according to $v$’’,
    i.e. $v_r$ is updated to
    the value vector $v^G_{\sigma'}$.
    This is precisely Line 6 of GSIA when we 
    update values in the subgame.
    \end{itemize}
    
    Hence, we see that except for the initialisation
    where $v_r$ may not correspond to a $\M$-strategy,
    it will be the case as soon as we reach Line 5 of
    the first loop, and from this point on CFB will
    correspond exactly to the instance of Opt-GSIA described above. 
    
        \begin{algorithm}
        	\caption{Converge From Below Algorithm\label{alg:CVB}}
        	\DontPrintSemicolon
        	\KwData{$G$ an SSG}
        	\KwResult{The optimal value vector $v^*$ of $G$}
        	\Begin{
        	    $\cdot$ let $v$ be a feasible vector in which all \m vertices
        	    have value $0$ and all \M vertices are stable \;
        	    \While{$v$ is not an optimal value vector}{
        	        $\cdot$ use linear programming to compute
        	        the value vector $v_r$ of 
        	        an optimal response 
        	        to the \M strategy that
        	        plays greedily according to $v$\;
        	        $\cdot$ update $v$ as the feasible vector where all \m vertices $x$
        	        have value $v_r(x)$ and all \M vertices
        	        are stable\;
        	    }
        	    
        	    \KwRet{$v$}
        	}
        \end{algorithm}
    
    \section{Algorithms Derived from GSIA} \label{section:salad}

    We show that all known strategy improvement algorithms can be expressed as instances of GSIA and we also propose several new algorithms, derived from choices of $A$ which make the transformed game polynomial time solvable. 
    The only algorithms which are not instances of GSIA are based on values rather than strategies:
    value propagation~\cite{Chatterjee2008valueiteration,condon1993algorithms,ibsen2012solving}, quadratic programming~\cite{condon1993algorithms,kvretinsky2020comparison} and dichotomy~\cite{auger2014finding}.
    
    \subsection{Hoffman-Karp Algorithms}
    
         The most classical method to solve an SSG, called the Hoffman-Karp algorithm, repeatedly applies switches to the strategy until finding the optimal one. It is also a generic algorithm, since the choice of the set of vertices to switch at each step is not specified nor the choice of the initial strategy. Many details on these algorithms can be found in~\cite{condon1993algorithms} or~\cite{tripathi2011strategy}.

        Hoffman-Karp algorithms are instances of GSIA, where $A$ is the set of all arcs of the SSG. Indeed, as proved in Lemma~\ref{lemma:switch}, a switch $\sigma'$ of $\sigma$ satisfies $\sigma \underset{G[A,\sigma]}{\succ} \sigma'$. Interpreting Hoffman-Karp algorithms as instances of GSIA proves that they work on non-stopping games, while in most article the stopping condition is required. Moreover, it shows that their number of iterations is $O(nq^r)$ on $q$-SSGs, a complexity exponential in $r$ only, which was known only for algorithms specially designed for this purpose~\cite{lmcs1119,dai2009new,ibsen2012solving,auger2019solving}.

        Ludwig's Algorithm~\cite{ludwig1995subexponential}, which is the best randomised algorithm to solve SSGs, can be seen as an Hoffman-Karp algorithm using Bland's rule as shown in~\cite{auger2019solving}: a random order on the vertices is drawn, and at each step, the first switchable vertex in the order is switched. Two other Hoffman-Karp algorithms are presented in~\cite{tripathi2011strategy}: switching all switchable vertices at each step or switching a random subset. Seeing these three algorithms as instances of GSIA yields $O(nq^r)$ as a deterministic bound on their number of iterations, which was unknown. However, the analysis of~\cite{ludwig1995subexponential,tripathi2011strategy} is required to obtain a good complexity in $n$ for these algorithms.

        \subsection{Selection of the Initial Strategy}

        In~\cite{dai2009new}, Dai and Ge give a randomised improvement of GHA simply by choosing a better initial strategy. To do so, they choose randomly $\sqrt{r!}\log(r!)$ strategies and choose the one with the highest value. This ensures, with high probability, that at most $\sqrt{r!}$ iterations will be done in GHA. Thus, their algorithm runs in $O(\left(\sqrt{r!}\right)$ iterations. This algorithm is also captured by GSIA by selecting the initial strategy in the same way, however it seems hard to combine the gain made by the random selection of the strategy and the bound in $O(q^r)$ of GSIA, since even a strategy close to the optimal one may have values far from it. Remark that it is trivial to extend this method to any instance of Opt-GSIA to improve on the complexity of Th.~\ref{theorem:fstrat}.

    \subsection{New Algorithms}
    
    We can use GSIA to design many strategy improvement algorithms. We present three of them, all based on a choice of $A$ which makes $G[A,\sigma]$ solvable in polynomial time. The initial strategy can be anything 
    and $\sigma'$ is always chosen to be the optimal strategy in $G[A,\sigma]$. Most of them can be seen as generalisations of known algorithms. 
    
    \begin{enumerate}
    \item  Let $A$ be a feedback arc set of $G$, then $G[A,\sigma]$ is acyclic and it can be solved in linear time. It seems intuitively appealing to think that this algorithm will be faster if the feedback arc set
    is small but we have no proof to sustain such a proposition.
    \item  A \M acyclic SSG is an SSG such that that every \M vertex has at most one outgoing arc in a cycle. \M acyclic SSG can be solved in polynomial time, see~\cite{auger2014finding}. If we let $A$ be a set of arc that contains all but one outgoing arcs of each \M vertex, then $G[A,\sigma]$ is \M acyclic and can be solved in polynomial time. Moreover, such a game can be solved by strategy improvement in at most $n$ iterations. This can be seen as a generalisation of Hoffman-Karp algorithm, in which $A$ contains \emph{all} outgoing arcs of \M vertices. 
    \item As an intermediate between acyclic games and \M acyclic games, we may consider almost acyclic games, where all vertices have at most one outgoing arc in a cycle. Almost acyclic SSGs can be solved in linear time~\cite{auger2014finding}.
    \end{enumerate}
 
\bibliography{algoseq.bib}

\end{document}